\documentclass[10pt]{amsart}

\usepackage{amssymb}
\usepackage{amsfonts}
\usepackage{amsmath}
\usepackage{algorithm}
\usepackage{algorithmic,color}

\usepackage{epsfig}
\usepackage{epsfig}

\newtheorem{theorem}{Theorem}[section]
\newtheorem{lemma}[theorem]{Lemma}
\newtheorem{corollary}[theorem]{Corollary}
\newtheorem{proposition}[theorem]{Proposition}
\newtheorem{problem}[theorem]{Problem}

\theoremstyle{definition}
\newtheorem{definition}[theorem]{Definition}

\theoremstyle{remark}

\usepackage{graphicx}
\numberwithin{equation}{section}
\def\R{\mathcal R}

\def\true{\mathrm {true} }
\def\false{\mathrm {false} }

\begin{document}

\title[Alternating Traps in Parity Games]{Alternating Traps in Muller and Parity Games}

\author[A. Grinshpun]{Andrey Grinshpun}
\address{Department of Mathematics\\
Massachusetts Institute of Technology\\
Cambridge, MA}
\email{agrinshp@mit.edu}

\author[P. Phalitnonkiat]{Pakawat Phalitnonkiat}
\address{Department of Mathematics\\
Cornell University\\
Ithaca, NY}
\email{pp287@cornell.edu}

\author[S. Rubin]{Sasha Rubin}
\address{IST Austria and TU Vienna\\
Austria}
\email{sasha.rubin@gmail.com}

\author[A. Tarfulea]{Andrei Tarfulea}
\address{Department of Mathematics\\
Princeton University\\
Princeton, NJ}
\email{tarfulea@princeton.edu}

\begin{abstract}
Muller games are played by two players moving a token along a graph; the winner is determined by the set of vertices that occur infinitely often. The central algorithmic problem is to compute the winning regions for the players. Different classes and representations of Muller games lead to problems of varying computational complexity. One such class are parity games; these are of particular significance in computational complexity, as they remain one of the few combinatorial problems known to be in NP $\cap$ co-NP but not known to be in P. We show that winning regions for a Muller game can be determined from the alternating structure of its traps. To every Muller game we then associate a natural number that we call its \textit{trap depth}; this parameter measures how complicated the trap structure is. We present algorithms for parity games that run in polynomial time for graphs of bounded trap depth, and in general run in time exponential in the trap depth.
\end{abstract}


\maketitle

\section{Introduction}

A Muller game \cite{M93}\cite{GTW02} is played on a finite directed graph in which the vertices are two-colored, say with colors red and blue. There is a token on an initial vertex and two players, call them Red and Blue, move the token along edges; it is Red's move if the token is on a red vertex, and otherwise it is Blue's move. To determine the winner, a Muller game also contains a collection $\R$ of sets of vertices. One assumes that there are no dead ends and so the play is an infinite walk. At each turn one records the vertex under the token. The winner is determined by the set $S$ of vertices that occur infinitely often; Red wins if $S$ is in $\R$, and otherwise Blue wins.

Every two-player perfect-information game with Borel winning condition is determined: one of the players has a winning strategy. In particular, every Muller game is determined: either Red or Blue has a winning strategy. To {\em solve a Muller game} is to determine for every vertex which player has a winning strategy when play starts from the given vertex. This set of vertices is called that player's {\em winning region}.

One application of these games is to solve Church's synthesis problem: construct a finite-state procedure that transforms any input sequence letter by letter into an output sequence such that the pair of sequences satisfies a given specification. The modern solution to this problem goes through Muller games \cite{T09}.

\subsubsection*{Characterization of Muller games}
The first part of this paper (section \ref{main-result}) characterizes the winning region of a Muller game $G$ in terms of a two player reachability game. The length of this reachability game is a measure of the alternating structure of the traps in $G$; we call it the {\em trap-depth} of $G$. We briefly explain.

Muller games admit natural substructures, \textit{Attractors} and \textit{Traps}.
The Red-attractor~\cite{Z98} of a subset $X$ of vertices is the set of vertices 
from which Red can force the token into $X$; this may be computed in linear time. A Red-trap~\cite{Z98} is a 
subset $Y$ of vertices in which Blue may keep the token within $Y$ indefinitely (no matter what Red does); i.e. if the token is in $Y$, Blue may choose to trap Red in the set $Y$. 
It should be evident that the complement of a Red-attractor is 
a Red-trap. Of course, all notions here (and elsewhere) defined for Red may be symmetrically defined for Blue.
Thus we talk of Blue-attractors and Blue-traps.

Now, consider the following game played on the same arena as a Muller game $G$. The {\em trap-depth game on $G$ in which Red goes first}  (Definition~\ref{trap_depth_game}) proceeds as follows (the traps discussed in the following are all nonempty): 
Red picks a Blue-trap $X_1 \subseteq V$ (here $V$ are the vertices of the Muller game $G$) which is winning for Red (i.e. $X_1 \in \R$). Then Blue picks a Red-trap $Y_1$ in the smaller game induced by $X_1$, where $Y_1$ is winning for Blue (i.e. $Y_1 \not \in \R$). Then Red picks a Blue-trap $X_2$ in the game induced by $Y_1$ such 
that $X_2$ is winning for Red. Red and Blue continue like this, alternately choosing traps. 
The first player that cannot move (i.e., that cannot find an appropriate nonempty trap) loses. As shown in Theorem~\ref{main_theorem},
\begin{quote}
Red has a nonempty winning region in the Muller game if and only if Red has a winning
strategy in the trap-depth game in which Red goes first.
\end{quote}
And if Red has a winning strategy in this trap-depth game, the first move of any 
winning strategy, $X_1$, contains only vertices in Red's winning region of the original Muller game.

\subsubsection*{Application to parity games}

The second part of the paper (section \ref{parityalg}) is algorithmic and applies the characterization of winning regions to a particular class of Muller games, parity games.

A parity game~\cite{EJ91} is played on a directed graph with vertices labeled by integers called {\em priorities}. This game is played between two players, Even and Odd, who move a token along edges. A vertex is called even if its priority is even, otherwise it is called odd.
Even moves when the token is on an even vertex, and Odd moves when the token is on an odd vertex. Play starts from a specific vertex; we assume there are no dead ends in the graph and so a play is an infinite walk. Even wins a play if the largest priority occurring infinitely often is even, otherwise Odd wins the play.

It is evident that parity games may be expressed as Muller games: the set $\R$ consists of all subsets $X$ of vertices in which the largest priority of vertices in $X$ is even.

Parity games are intertwined with a logical problem:
the model checking problem for Modal $\mu$-calculus formulas is log-space equivalent to solving parity games ~\cite{GTW02}. 
Complexity-wise, the problem is known to be in $NP \cap \text{ co-} NP$~\cite{EJS93}, and even 
$UP \cap \text{ co-} UP$ ~\cite{J98}: one of the few combinatorial problems in that category that is not known to be in 
$P$.\footnote{
Note that for purposes of computational complexity, the size of a parity game includes the size of the graph plus some considerations on the size of the integers, but we ignore this latter point.}


The algorithmically-minded reader may observe a potential drawback with reinterpreting games as a game of alternating traps. The number of traps in a 
game can grow exponentially with the size of the game (just take a graph with only self-loops), and what's worse is that we are looking at
chains of alternating traps. Nonetheless, we apply the characterization to parity games: say that a graph has {\em  Even trap-depth at most $k$} if Even can guarantee that, in the trap-depth game in which Even goes first, the game ends in a win for Even within $k$ rounds. 
Then, despite the previous observation, we present an algorithm TDA$(G,\sigma,k)$ (here $G$ is a parity game, $\sigma$ is a player, and $k$ an integer)
that runs in time $|G|^{O(k)}$ and, as shown in Theorem \ref{theorem_TDA}, 
\begin{quote}
returns the largest (possibly empty) set starting with which $\sigma$ can guarantee a win in at most $k$ moves in the trap-depth game on $G$.
\end{quote}
Note that the definition of trap depth may be applied to Muller games as well, though we do not have an algorithmic application; one might hope that there are particularly efficient algorithms for finding winning vertices in Muller games of small trap depth.

Let's put this all together. Say that a parity game has {\em trap-depth at most $k$} if either it has Even trap-depth at most $k$ or Odd trap-depth at most $k$. In Figure \ref{Figure 2}
we exhibit, for every integer $k$, a parity game with $O(k)$ vertices and edges that has trap-depth exactly $k$. By the end of the paper we will have algorithmically solved the following problems:
\begin{enumerate}
\item decide if a given parity game $G$ has trap-depth at most $k$.
\item find a nonempty subset of one of the player's winning region assuming the game has trap depth at most $k$.
\end{enumerate}
Moreover, these problems can be solved in time  $O(mn^{2k-1})$ where $n$ is the number of vertices and $m$ the number of edges of a parity game $G$. 

\section{Muller Games and Parity Games} \label{sec:Mullergames}

\label{preliminaries}

A \emph{Muller Game} $G=(V,V_{\mbox{red}}, E,\R)$ satisfies the following conditions: $(V,E)$ is a directed graph in which every vertex has an outgoing edge, $V$ is partitioned into {\em red vertices} $V_{\mbox{red}}$ and {\em blue vertices} $V_{\mbox{blue}}:= V \setminus V_{\mbox{red}}$, and $\R \subset 2^V$ is a collection of subsets of $V$. The Muller game is played between two players, Red and Blue.
Red will move when the token is on a red vertex, and otherwise Blue will move. Starting from some vertex $v_0$, Blue's and Red's moves result in an infinite sequence of vertices, called a {\em play}, $P=(v_0,v_1,v_2,\ldots)$ where $(v_i,v_{i+1}) \in E$. Taking inf$(P)$ to be the set of vertices that occur infinitely often in the play, i.e. $v \in \text{inf}(P)$ if and only if there are infinitely many $i$ so that $v_i=v$, we say Red  {\em wins the play} if inf$(P) \in \R$, and otherwise Blue wins the play.


 
Take $\sigma \in \{\mbox{Red},\mbox{Blue}\}$ (we write $\overline \sigma$ for the other player, so if $\sigma$ is Red then $\overline{\sigma}$ is Blue, and vice versa). A \emph{$\sigma$-Strategy} is an instruction giving Player ${\sigma}$'s next move given the current token position and play history. Formally, it is a function 
whose domain is the set of finite strings of vertices $\{v_0 v_1 \cdots v_k$ : $(v_i,v_{i+1}) \in E \}$ and whose range 
is $N(v_k) := \{ v \in V$ : $(v_k,v) \in E \}$, the neighborhood of $v$. A $\sigma$-strategy is \emph{winning from vertex $v_0$} if, for all plays starting at $v_0$ and for which that strategy is followed whenever it is $\sigma$'s turn, the 
resulting play is winning for $\sigma$. Finally, a $\sigma$-strategy is \emph{memoryless} if it gives $\sigma$'s move 
while taking into consideration only the current token position; i.e., it is a strategy in which the value on $v_0 \cdots v_k$ depends only on $v_k$.
A given memoryless $\sigma$-strategy $\pi$ in a Muller game $G$ induces a subgame $H$ in which we restrict the outgoing edges of any $\sigma$ vertex to the edge defined by $\pi$.
It is worth noting that if both players fix a strategy for the game, then the resulting play is completely determined by the starting vertex, since given the current history we can determine which vertex is visited next.

Muller games are determined (since they are Borel we can apply \cite{M75}, although for the special case of regular games see \cite{GTW02}): starting from any vertex, there is a player that has a winning strategy. Determinacy partitions $V$ into the respective \emph{winning regions} $W_{G,\text{Red}}$ and $W_{G,\text{Blue}}$ (where $v \in W_{G,\sigma}$ if and only if $\sigma$ has a winning strategy starting from $v$ in $G$). In contexts where the meaning is clear, we will use $W_\sigma$ for $W_{G,\sigma}$. It follows easily that for a player $\sigma$ there is a single strategy that wins starting from any vertex in $W_{G,\sigma}$; such a strategy is called a winning strategy. We now introduce various important substructures of Muller games that capture some of the essential concepts of reachability and restriction (see chapter $2.5$ of \cite{GTW02}).

\begin{definition}
A \emph{$\sigma$-Trap} is a collection of vertices $X \subseteq G$ where:
$$\forall x \in X \cap V_{\sigma} \text{ we have } N(x) \subseteq X$$ and
$$\forall x \in X \cap V_{\overline\sigma},\ \exists y \in X\text{ such that } (x,y) \in E.$$
\end{definition}

No $\sigma$-vertex in $X$ has an outgoing edge leaving the trap, and every $\overline\sigma$-vertex in $X$ has at least one 
outgoing edge that stays in the trap. Consequently, if the token ever enters $X$, $\overline\sigma$ has a strategy through which the token will never leave $X$, no matter what $\sigma$ does. It is apparent 
that $W_{\sigma}$ is a $\overline\sigma$-trap. 

{\em Notation.} We write ${Traps}_{\sigma}(G)$ to denote the set of nonempty  $\sigma$-traps in $G$.

\begin{definition}
 A \emph{$\sigma$-Attractor of a set of vertices $Y$} is the set of vertices starting from which $\sigma$ has a strategy that guarantees
$Y$ will be reached (after finitely many, possibly $0$, steps).
\end{definition}

We denote the attractor of a set $X$ in a graph $G$ with respect to a player $\sigma$ by $Attr(G,X,\sigma)$, and it is worth noting that the attractor of a set may be computed in time linear in the size of the graph; the algorithm for doing so is presented below~\cite{Z98}.

\begin{algorithm}

\caption{Attr$(G=(V,E,p),X,\sigma)$}

\begin{algorithmic}[1]

\STATE $C_{\text{prev}}:= \emptyset$
\STATE $C_{\text{cur}} := X$
\WHILE{$C_{\text{cur}} \neq C_{\text{prev}}$}
\STATE $C_{\text{prev}}:=C_{\text{cur}}$
\STATE $C_{\text{cur}}:=C_{\text{prev}} \cup \{v \in V_\sigma : N(v) \cap C_{\text{prev}} \neq \emptyset\} 
\cup \{v \in V_{\overline\sigma} : N(v) \subseteq C_{\text{prev}}\}$
\ENDWHILE
\RETURN $C_{\text{cur}}$
\end{algorithmic}
\label{alg1}
\end{algorithm}

On each iteration, the $\sigma$ vertices that have an edge into the part of the attractor that has already been computed are added, and the $\overline\sigma$ vertices that have only edges into that part are added. We briefly argue correctness: by induction on the number of iterations, we see that starting anywhere in the computed set, $\sigma$ has a strategy to reach $X$, and starting outside the computed set it is easy to see that $\overline\sigma$ has a strategy to avoid the computed set indefinitely (every $\overline\sigma$ vertex outside the set has some edge that does not enter the set, and every $\sigma$ vertex outside of it has no edge that enters it), so this does compute the attractor.

\begin{definition}
The \emph{Induced Subgame of $G$ by $X$} is the Muller Game using the vertices $V \cap X$ and the edges 
$E \cap X^2$; we sometimes refer to this as ``$G$ restricted to $X$'' and use the notation $G[X]$.
\end{definition}

Naturally, $G[X]$ should have no dead-ends if it is to be a Muller game. It is apparent that $G$ restricted to a trap $X$ is a Muller game. When $X$ is a trap we use {\em subtraps} to mean the traps of $G[X]$.

\begin{lemma} \cite{Z98}
\label{reducibility}
If $X \subseteq W_{G,\sigma}$ then, taking $U = Attr(G,X,\sigma)$, we have $W_{G,\sigma} = U \cup W_{G[V \setminus U], \sigma}$.
\end{lemma}

In other words, if we know that $\sigma$ can win from a set, then we can remove that set's attractor from the graph and just find the winning region for $\sigma$ in the smaller graph.


\begin{lemma} \cite{Z98}
\label{traps}
If $X \subseteq W_{G,\sigma}$ and $X$ is a $\sigma$-trap, then $W_{G[X],\sigma}=X$.
\end{lemma}

Intuitively, this holds because in the induced game $G[X]$ player $\sigma$ can continue to use the same winning strategy that $\sigma$ had in $G$.

We end the section with the statements of some technical lemmas that will be useful. Their proofs are routine.


\begin{lemma} \cite{Z98}
If $X$ is a $\sigma$-trap in $G$ and $Y$ is a $\sigma$-trap in $G[X]$, then $Y$ is a $\sigma$-trap in $G$.
\end{lemma}


The next lemma states that if we take the $\sigma$ attractor of some set $Y$ and are interested in how it intersects with some $\sigma$-trap $X$, then the intersection is contained in the attractor of $X \cap Y$ in the game restricted to $X$.

\begin{lemma} \cite{Z98}
If $X$ is a $\sigma$-trap in $G$, $Y$ is a set of vertices, and $S = \text{Attr}(G,Y,\sigma)$, then $X \cap S \subseteq \text{Attr}(G[X],X \cap Y, \sigma)$.
\end{lemma}

\begin{lemma}
If $X$ is a $\sigma$-trap in $G$ and $Y$ is a $\overline\sigma$-trap in $G$, then $X \cap Y$ is a $\overline\sigma$-trap in $G[X]$.
\end{lemma}



\subsection{Parity games}

A \emph{Parity Game} $G=(V,E,\rho)$ satisfies the following conditions: $(V,E)$ is a directed graph in which every 
vertex has an outgoing edge, $v_0\in V$ denotes a starting vertex, and $p:V\rightarrow \mathbb{Z}$ is a function assigning priorities to the vertices. 
The parity game is played between two players, Even and Odd, where each player moves the token along a directed edge of $G$ whenever the token is on a vertex of the corresponding 
parity. We say a vertex is even if it has even priority and odd if it has odd priority.
Even's and Odd's moves result in an infinite play:
$P = (v_0,v_1,v_2,...)$
where $(v_i,v_{i+1}) \in E$. Even wins the play if $\limsup_{i \in \mathbb{N}} p(v_i)$ is even and Odd wins otherwise: 
i.e., the largest priority that occurs infinitely often determines the winner of the play.


Note that, given a parity game, we may define the corresponding Muller game by placing $v$ in $V_{\mbox{red}}$ if and only if $p(v)$ is even. Then $S \subseteq V$ has $S \in \R$ if and only if $\max(S)$ is even, and otherwise $\max(S)$ is odd and $S \not \in \R$. The corresponding Muller game is then $(V,V_{\mbox{red}},E,\R)$. Note that a play is winning in the Muller game if and only if it is winning in the parity game.

Not only are Parity games determined, they are \emph{Memorylessly Determined} ~\cite{EJ91}: for every vertex $v \in V$, exactly 
one of the two players has a memoryless strategy that guarantees a win starting from $v$. Moreover, for each player there is a single memoryless strategy which, if followed, will result in a winning play starting from any vertex in that player's winning region; this is a called a memoryless winning strategy. Note that Muller games are not memorylessly determined; they may require a strategy that uses some of the play history.

\section{The Trap-Depth Game}
\subsection{Main Theorem}
\label{main-result}
As mentioned in the introduction, our main result relies on a characterization stemming from chains of alternating subtraps. Each subtrap represents the decision of the corresponding player to further restrict the token's 
movement. This goes on until the final restriction leaves one player incapable of preventing a winning play for their 
opponent. We now formalize this idea. We begin by defining a set of statements related to chains of 
alternating traps. 

Define $R_\sigma$ to be $\R$ if $\sigma$ is Red and $2^V \setminus \R$ otherwise. The statement $S \in R_\sigma$ says that if the set of vertices that occurs infinitely often is $S$, then player $\sigma$ wins. Recall that ${Traps}_{\sigma}(G)$ is the set of \emph{nonempty} $\sigma$-traps in $G$. Our boolean statements $\Delta_\sigma(G,k)$ are defined recursively and have three parameters:  the player $\sigma$, the game $G$, and the iteration (or depth) number $k$.

\begin{definition}
For player $\sigma$, game $G$, and integer $k$, the value of $\Delta_{\sigma}(G, 0)$ is $\false$. For $k > 0$, the value of $\Delta_{\sigma}(G, k)$ is $\true$ if and only if there exists $X \in {Traps}_{\overline\sigma}(G)$ such that
\begin{itemize}
\item $X \in R_\sigma$, and
\item $\forall Y \in {Traps}_{\sigma}(G[X]) \text{ we have } Y \in R_\sigma \text{ or }  \Delta_{\sigma}(G[Y], k-1)$.
\end{itemize}
\end{definition}
Each statement $\Delta_{\sigma}(G, k)$ asserts that $\sigma$ can restrict the token's movement via a trap $X$ 
in such a way that if every vertex in the trap occurs infinitely often, player $\sigma$ wins, i.e. $X \in R_\sigma$, (intuitively then, player $\overline{\sigma}$ must choose to further restrict play) and, no matter how $\overline\sigma$ further restricts 
the token's movement via a subtrap $Y$, either still $Y \in R_\sigma$ 
or we have that $\Delta_{\sigma}(G[Y],k-1)$ is true. So, in particular, $\Delta_{\text{Red}}(G, 1)$ states that there is a 
Blue-trap $X$ in $G$ with $X \in \R$ such that every Red-subtrap $Y$ has $Y \in \R$.

The above definitions make it easy to see that the statements make references to natural structures in Muller games, but they can be rather cumbersome to work with, so we present an equivalent but easier to visualize way to think about them.

\begin{definition} \label{trap_depth_game}
Let $G$ be a Muller game. Define the \textit{Trap-Depth Game on G in which $\sigma$ goes first} as follows: 
in the beginning of the $i^\text{th}$ round ($i \geq 1$)
there will be some current Muller game $G_i$. The game starts with $G_1 = G$. In the $i^\text{th}$ round player $\sigma$ moves
first by choosing a trap $X_i \in {Traps}_{\overline\sigma}(G_i)$ with $X_i \in R_\sigma$. Player $\overline\sigma$ replies by choosing a $\sigma$-trap $Y_i$ in the subgame $G_i[Y_i]$, i.e. $Y_i \in {Traps}_\sigma(G_i[X_i])$, so that $Y_i \in R_{\overline\sigma}$. This completes the $i^{\text{th}}$ round. Define $G_{i+1}=G_i[Y_i]$. The first player that has no legal move loses.
\end{definition}

In a Muller game, this will terminate in at most $\left\lceil \frac{n}{2} \right\rceil$ rounds, as each time a player chooses a trap, a vertex must be removed. If the Muller game is a parity game, then the condition $X \in R_\sigma$ simply states that the largest priority of a vertex in $X$ is of parity $\sigma$. For a parity game, the number of rounds is at most  $\left\lceil \frac{|p(V)|}{2} \right\rceil$, since the size of the largest vertex still in play decreases twice per round. In particular, every play in this game is finite and ends in a win for one of the players. Therefore, the game is determined (i.e. one of the players has a winning strategy).

\begin{lemma}
The value of $\Delta_{\sigma}(G,k)$ is $\true$ if and only if $\sigma$ has a strategy that ensures their opponent loses the 
Trap-Depth Game in which $\sigma$ goes first in at most $k$ rounds (so $\overline\sigma$ would lose on or before the $2k^{\text{th}}$ move).
\end{lemma}
This is easily verified by identifying player moves with the quantifiers in the expression for $\Delta_{\sigma}(G,k)$. 
We now arrive at the first main result of this paper:

\begin{theorem}
\label{main_theorem}
Let $G$ be a Muller game. Then $W_{G,\sigma} \neq \emptyset$ if and only if $\sigma$ has a winning strategy in the trap-depth game on $G$ in which $\sigma$ goes first. Moreover, any first move $X$ in a winning strategy by $\sigma$ satisfies $X \subseteq W_\sigma$.
\end{theorem}
So Player ${\sigma}$ has some nonempty winning region in the game $G$ if and only if $\sigma$ has a winning strategy in the 
Trap-Depth Game in which $\sigma$ goes first.

Note the following simple corollary:
\begin{corollary}
The following two statements are equivalent:
\begin{itemize}
\item Parity games can be solved in polynomial time.
\item The player with a winning strategy in the trap-depth game described by a parity game can be determined in polynomial time.
\end{itemize}
\end{corollary}

This theorem also motivates a new parameter for parity games:

\begin{definition}
The \emph{Trap-Depth} of a parity game $G$ is the minimum integer $k$ such that $\Delta_{\mbox{Even}}(G,k)$ or $\Delta_{\mbox{Odd}}(G,k)$. 
\end{definition}

Note that this is a parameter that fundamentally depends on both the graph and the priorities of the vertices. Although having bounded trap-depth is much more general, one simple class of parity games that has this property is those with a bounded number of priorities.

The above definition applies equally well to Muller games, though we do not have an algorithmic application. Similarly, one can define the \emph{$\sigma$-trap-depth} of $G$ as the minimum integer $k$ (if it exists) such that 
$\Delta_{\sigma}(G,k)$; so $W_\sigma \neq \emptyset$ 
if and only if the $\sigma$-trap depth of $G$ is at most $\left\lceil \frac{|p(V)|}{2} \right\rceil$. This upper bound can be achieved, as shown by Figure \ref{Figure 2}.
\begin{figure}[hbt]
  \centering
    {\psfig{figure=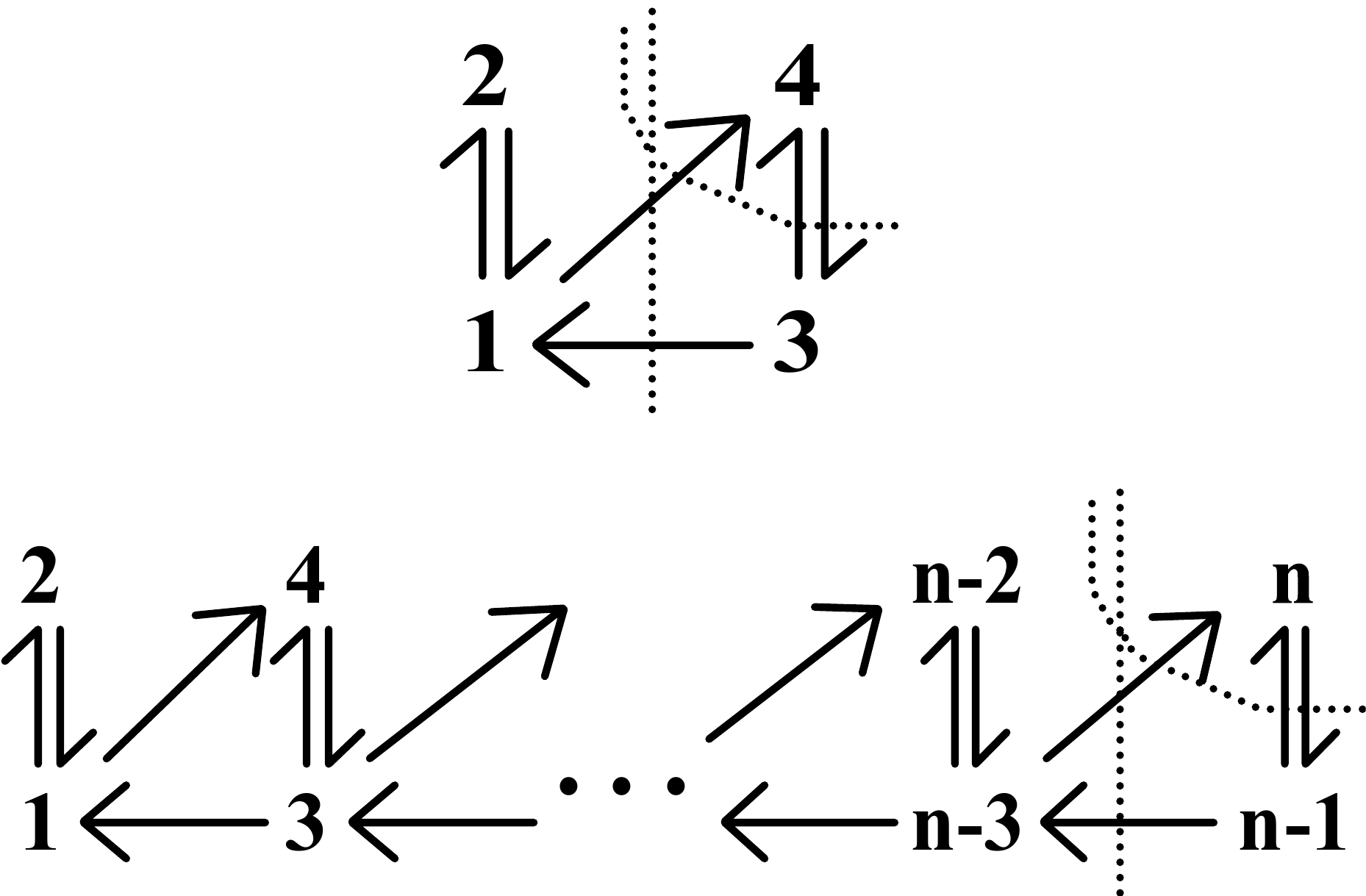,
    width=0.45\textwidth}} 
\caption{{\bf Maximum Trap-Depth:} Above: base case ($G_4$) with trap-depth $2$; Below: $G_{n}$ with $n$ vertices ($n$ is even); both are Even-winning from every vertex (so $\Delta_1(G_n,k)$ is never true for any $k$, by Theorem~\ref{main_theorem}). 
The only Odd-trap is the entire graph, so this must be Even's first move in a trap-depth game. Odd could then remove the right-most top vertex, 
the remaining set being an Even-subtrap. Within this graph, the only remaining appropriate Odd-subtrap 
for Even's next move is the set formed by removing the right-most bottom vertex. We have now 
reduced the game to $G_{n-2}$. Each time we add two vertices we increase the trap-depth by $1$, so the trap depth of $G_{n}$ is exactly $n/2$}
\label{Figure 2}
\end{figure}

\subsection{Proof of Theorem~\ref{main_theorem}}
\label{proof-main-theorem}
\subsubsection{Proof for Memoryless Strategies}

We will first prove the characterization of Muller games (the first two sentences of Theorem \ref{main_theorem}) for games in which player $\sigma$ has a memoryless strategy that wins starting from any vertex in $W_\sigma$. Intuitively, traps do not distinguish between memoried and memoryless strategies; we will formalize this intuition and this will allow us to extend the main theorem to all Muller games.

\begin{lemma} \label{lemma_maxismine}
Let $G$ be a nonempty Muller game with $W_{G,\sigma}=V$, that is in which $\sigma$ wins starting from any vertex, and $\pi$ a memoryless winning strategy for $\sigma$. Then there is a nonempty $\overline\sigma$-trap $T$ in $G$ such that
$W_{G[T],\sigma}=T$, $T \in R_\sigma$, and, if $\pi$ is followed, then any play starting in $T$ will not leave $T$ (i.e. $\pi$ does not prescribe leaving $T$).
\end{lemma}
\begin{proof}
Fix a memoryless winning strategy $\pi$ for $\sigma$ in $G$, and take $H$ to be the subgame induced by $\pi$; that is, leave only one edge out of each $\sigma$ vertex, the one corresponding to the strategy $\pi$. Take $T$ to be a strongly connected component (SCC) of $H$ such that $T$ has no edges into any other SCC. Note that $T$ is a $\overline\sigma$-trap in $H$, and so also in $G$. Since $T$ is strongly connected and player $\sigma$ only has one possible move at any vertex, $\overline\sigma$ has a strategy (not necessarily memoryless) such that starting from any vertex in $T$, if the strategy is followed, every vertex in $T$ occurs infinitely often. Then, by the assumption that $\pi$ was winning, we must have $T \in R_\sigma$. By construction, $\pi$ does not prescribe leaving $T$.
\end{proof}

The following two propositions establish the theorem for Muller games in which $\sigma$ has a memoryless winning strategy.

\begin{proposition}\label{prop1}
If $W_{G,\sigma} \neq \emptyset$ and $\sigma$ has a memoryless winning strategy, then $\sigma$ has a winning strategy on the trap-depth game on $G$ in which $\sigma$ goes first.
\end{proposition}
\begin{proof}

Fix $\pi$ a memoryless winning strategy for $\sigma$. We describe a strategy for $\sigma$ in the trap-depth game so that for every $i \geq 1$ player $\sigma$ has a valid move $H_i$ satisfying that $\pi$ does not prescribe leaving $X_i$ and any potential response $Y_i$ satisfies the invariant $W_{G[Y_i],\sigma}=Y_i$. To get the induction going we define $Y_0:=W_{G,\sigma}$ and note that $W_{G[Y_0],\sigma}=Y_0$. Note that such a strategy ensures that player $\sigma$ always has a valid move and thus wins the trap-depth game.

Suppose $i \geq 0$ rounds have been played, and assume by induction that $\sigma$ wins the Muller game starting from any vertex in $Y_i$. Then, by Lemma \ref{lemma_maxismine}, there is some $\overline\sigma$-trap $X_{i+1}$ in $G[Y_i]$ with $X_{i+1} \in R_\sigma$ such that $W_{G[X_{i+1}],\sigma}=X_{i+1}$ and $\pi$ does not prescribe leaving $X_{i+1}$; have $\sigma$ play such an $X_{i+1}$. Then, if player $\overline\sigma$ has some response $Y_{i+1}$, we have that $Y_{i+1}$ is a $\sigma$-trap in $G[X_{i+1}]$ and so by Lemma~\ref{traps} $W_{G[Y_{i+1}],\sigma}=Y_{i+1}$, as required.
\end{proof}

\begin{proposition}\label{prop2}
If $W_{G,\sigma} = \emptyset$ and player $\overline\sigma$ has a memoryless winning strategy, then $\overline\sigma$ has a strategy that wins the trap-depth game on $G$ in which $\sigma$ goes first.
\end{proposition}
\begin{proof}
Let $X$ be player $\sigma$'s first move. Then, since $X$ is a $\overline\sigma$-trap, we have $W_{G[X],\overline\sigma} = X \neq \emptyset$ by Lemma~\ref{traps}. Note that now we simply play the trap-depth game on $G[X]$ in which $\overline\sigma$ goes first and $\pi|_X$ is a memoryless winning strategy on $G[X]$, and so by the previous proposition we have that $\overline\sigma$ has a winning strategy.
\end{proof}

The previous two propositions show the desired characterization of Muller games, assuming that players have memoryless winning strategies.

\subsubsection{Proof for all Muller Games}

While Muller games do not, in general, have memoryless strategies, a player need only use a finite amount of memory. To formalize this notion, we define a bounded-state strategy.

\begin{definition}
For any Muller game $G$, any positive integer $N$, and any function $M: V \times [N] \rightarrow [N]$, define the $M$-sequence with respect to any play $v_0,v_1,\ldots$ by $M_0 = 0$ and $M_i = M(v_i,M_{i-1})$.
\end{definition}

Intuitively, in the above, $M_i \in [N]$ is the (joint) memory used by the players and $M_{i+1}$ depends only on $M_i$ and on the most recent move.

\begin{definition}
For any Muller game $G$, any positive integer $N$, and any function $M: V \times [N] \rightarrow [N]$, a strategy $\pi_\sigma$ for player $\sigma$ is a bounded-state $M$-strategy if there is some $\pi : [N] \rightarrow V$ so that if $v_0,\ldots$ is any play consistent with $\pi_\sigma$ and $M_0,\ldots$ is the corresponding $M$-sequence, then $\pi_\sigma$ depends only on $M_i$. I.e., there is some function $\pi : [N] \rightarrow V$ so that for each $\sigma$-vertex $v_i$, we have $v_{i+1}=\pi(M_i)$.
\end{definition}

The following theorem is proved in \cite{M93}. It states that, while Muller games may not have memoryless strategies, players need only a bounded amount of memory.

\begin{theorem}
For any Muller game $G$ there is some positive integer $N$ and some $M: V \times [N] \rightarrow [N]$ so that, for each player $\sigma$, there is a bounded-state $M$-strategy $\pi_\sigma$ satisfying that, starting from any vertex in $\sigma$'s winning region, $\pi_\sigma$ is a winning strategy for $\sigma$.
\end{theorem}

Given any Muller game $G$, take $N$ a positive integer and $M: V \times [N] \rightarrow [N]$ as in the previous theorem. We define the memoried Muller game associated with $G$, call it $G_M$, to have vertex set $V\times [N]$ (where $N$ depends on $G$ as in the previous theorem). Intuitively, $G_M$ will simulate $G$, but each vertex $v \in V \times [N]$ in the memoried game records the current state of the memory, with $v_1$ representing the current position in $G$. Thus, given $v,w$ vertices in the memoried game, $(v,w)$ is an edge of the memoried game if and only if $(v_1,w_1)$ is an edge in $G$ and $w_2=M(w_1,v_2)$ (here $v_1 \in V$ is the first coordinate of $v$ and $v_2 \in [N]$ is the second coordinate). Define the vertices belonging to player $\sigma$, $V_{\sigma,M}$, by $v \in V_{\sigma,M}$ if and only if $v_1 \in V_\sigma$. Similarly, $S \subseteq V^n$ is winning for Red, i.e. has $S \in \R_M$, if and only if the corresponding vertices are winning for Red in the original Muller game $G$, i.e. if and only if $\{v_1 : v \in S\} \in \R$.

Note that by the previous theorem and the construction of the memoried games, both players have memoryless winning strategies in $G_M$. The remainder of this section argues that the trap-structure of $G_M$ is very similar to that of $G$.

Intuitively, the following lemma says that if, when playing the trap-depth game on $G_M$, player $\overline{\sigma}$ simply pretends it's the trap-depth game on $G$, then any edge out of a $\overline{\sigma}$ vertex that would have existed were the game played on $G$ also exists in the game on $G_M$.

\begin{lemma}
Assume, in a trap-depth game on $G_M$ whenever the current set of vertices is $X_M$ and it is $\overline{\sigma}$'s turn to move, that $\overline{\sigma}$'s move has the following form: taking $X:=\{v_1 : v \in X_M\}$, there is some $\sigma$-trap $Y$ in $G[X]$ so that $\overline{\sigma}$'s move is $X_M \cap (Y \times [N])$. Then, at every point in the game, if the current set of vertices is $X_M$, take $X:=\{v_1 : v \in X\}$. For any $v \in X_M$ with $v_1 \in V_{\overline{\sigma}}$ and for any $u \in X$ so that $(v_1,u)$ is an edge of $G$, there is some $w \in X_M$ so that $(v,w)$ is an edge of $G_M$ and $w_1=u$.
\end{lemma}
\begin{proof}
We proceed by induction on the number of plays in the game. In the base case, the game is the whole graph and this is true by construction of $G_M$. Take $X_M$ to be the current set of vertices and $X:=\{v_1 : v \in X_M\}$.

If it is $\sigma$'s turn to move, $\sigma$ chooses some $\overline\sigma$ trap $Y_M$. Taking $Y:=\{v_1 : v \in Y_M\}$, for any $v \in Y_M$ with $v_1 \in V_{\overline{\sigma}}$ and for any $u \in X$ so that $(v_1,u)$ is an edge of $G$, by induction there is some $w \in X_M$ so that $(v,w)$ is an edge of $G_M$ and $w_1=u$. But $Y_M$ is a $\overline\sigma$ trap, so since $v \in Y_M$ and $v$ is a $\overline{\sigma}$ vertex we get $w \in Y_M$.

If it is $\overline\sigma$'s turn to move, $\overline\sigma$ chooses some $\sigma$ trap $Y_M$ of the form $X_M \cap (Y \times [N])$ where $Y$ is a $\sigma$ trap in $X$. Note that $Y=\{v_1 : v \in Y_M\}$, so this notation is consistent with previous notation. Given any $\overline\sigma$ vertex $v \in Y_M$ and any $u \in Y$ so that $(v_1,u)$ is an edge of $G$, by induction there must be some $w \in X_M$ with $w_1 = u$ so that $(v,w)$ is an edge of $G_M$. But then $w \in Y_M$ since $w_1=u \in Y$ and $Y_M = \{w \in X : w_1 \in Y\}$.
\end{proof}

\begin{theorem}
Player $\sigma$ has a winning strategy in a trap-depth game (in which either $\sigma$ or $\overline{\sigma}$ goes first) on $G$ if and only if player $\sigma$ has a winning strategy in a trap-depth game on $G_M$ (in which the same player goes first).
\end{theorem}
\begin{proof}
Assume player $\sigma$ has a winning strategy in a trap-depth game on $G_M$. Then player $\sigma$ is to play a trap-depth game on $G$ and we wish to show that player $\sigma$ has a winning strategy; we define player $\sigma$'s strategy by emulating the game on $G_M$. With each move, player $\sigma$ will maintain a set of vertices $X_M$ which represents the state of the emulated game on $G_M$. Assume the current set of vertices in the game on $G$ is $X$, and $\sigma$ has maintained the state $X_M$. We will inductively show that $\sigma$ has a strategy that maintains $X = \{v_1 : v \in X_M\}$ and that, starting from $X_M$ with the appropriate player moving, is winning for $\sigma$ in the game $G_M$. In the base case, $X=V$ and $X_M=V_M$.

If it is $\overline{\sigma}$'s turn to move, $\overline{\sigma}$ will pick some $\sigma$ trap $Y \in R_{\overline\sigma}$. Then we claim $Y_M:= X_M \cap \left(Y \times [N]\right)$ is a $\sigma$-trap in $X_M$: since $Y$ is a $\sigma$-trap in $X$, it must be the case that given any $\sigma$ vertex in $Y_M$, any neighbor it had in $X_M$ is also in $Y_M$. Given a $\overline\sigma$ vertex $v$ in $Y_M$, since $Y$ is a $\sigma$-trap we have that $v_1$ has a neighbor in $Y$, and so $v$ has a neighbor $w$ in $Y_M$ by the previous lemma, thus verifying that $Y_M$ is a $\sigma$-trap. Then $Y = \{v_1 : v \in Y_M\}$ so $Y_M \in R_{\overline\sigma,M}$ since $Y \in R_{\overline\sigma}$. Since $X_M$ was winning for $\sigma$, we have $Y_M$ is as well (since any move by $\overline{\sigma}$ must result in a winning position).

If it is $\sigma$'s turn to move, by assumption we are in some winning position $X_M$. Then $\sigma$ may choose some $\overline\sigma$ trap $Y_M \in R_{\sigma,M}$ in $G_M[X_M]$ so that $Y_M$ is winning for $\sigma$. We claim $Y:= \{v_1 : v \in Y_M\}$ is a $\overline\sigma$ trap in $X$. Since $Y_M$ is a $\overline\sigma$ trap in $X_M$, given any $\sigma$ vertex $u \in Y$ choose $v \in Y_M$ with $v_1=u$; $v$ must have some neighbor $w \in Y_M$, so $(v_1,w_1)$ is an edge of $Y$. Given any $\overline\sigma$ vertex $t \in Y$ and any neighbor $u \in Y$, we may choose $v \in Y_M$ with $v_1=t$ and then we have that there is some $w \in X_M$ which is a neighbor of $v$ with $w_1=u$ by the previous lemma, but $Y_M$ is a $\overline\sigma$-trap, so $w \in Y_M$ and so $u = w_1 \in Y$, as desired.

We've shown that, if $\sigma$ has a winning strategy on $G_M$, then $\sigma$ has a winning strategy on $G$. Symmetrically, if $\overline{\sigma}$ has a winning strategy on $G_M$, then $\overline{\sigma}$ has a winning strategy on $G$, thus proving the theorem.
\end{proof}

By combining the previous theorem with Propositions \ref{prop1} and \ref{prop2}, we may remove the assumptions regarding having memoryless winning strategies:

\begin{theorem}
If $W_{G,\sigma} \neq \emptyset$, then player $\sigma$ has a winning strategy on the trap-depth game on $G$ in which $\sigma$ goes first.
\end{theorem}

\begin{theorem}
If $W_{G,\sigma} = \emptyset$, then player $\overline\sigma$ has a winning strategy on the trap-depth game on $G$ in which $\sigma$ goes first.
\end{theorem}

Assume $T$ is the first-move $\overline\sigma$-trap in the trap-depth game on $G$ where $\sigma$ goes first and that $\sigma$ wins starting from $G[T]$ if $\overline{\sigma}$ goes first. If $X := T \cap W_{\overline\sigma} \neq \emptyset$, then $X$ is a $\sigma$-trap in $G[T]$ with $W_{G[X],\overline\sigma}=X$. So, by Lemma \ref{lemma_maxismine}, if we consider $X_M$ in $G_M$, $\overline\sigma$ has a viable move $Y_M \subseteq X_M$ such that $W_{G_M[Y_M],\overline\sigma}=Y_M$. By our previous arguments we then get that $\overline{\sigma}$ can win in the trap depth game in which $\sigma$ goes first on $G[Y]$ where $Y = \{v_1 : v \in Y_M\}$ is a $\sigma$-trap in $G[X]$. Then $Y$ is a valid move for $\overline\sigma$ in $G[T]$, contradicting the assumption that $\sigma$ can win from $G[T]$ if $\overline\sigma$ goes first.
\begin{corollary}
If $T$ is the first-move of a winning $\sigma$-strategy in the Trap-Depth Game where $\sigma$ goes first,
then $T \subseteq W_{\sigma}$.
\end{corollary}

This completes the proof of Theorem \ref{main_theorem}.

It is interesting to understand how these nested traps will interact with modifications to the graph. The following theorem says that via one such modification not much information is lost; this is particularly useful if one wishes to run the algorithms discussed in the next section.

\begin{theorem}\label{intersectTrap}
Let $G$ be a Muller game. Assume that, in the trap depth game on $G$, $X$ is a valid first move for $\sigma$ that allows $\sigma$ to guarantee a win in at most $k$ rounds and that $A$ is a $\sigma$-trap so that $A \cap X$ is non-empty. Then there is $Y \subseteq X \cap A$ where $Y$ is a valid first move for $\sigma$ that allows $\sigma$ to win in at most $k$ rounds on the trap depth game on $G[A]$.
\end{theorem}
\begin{proof}
Since $X$ is a $\overline{\sigma}$ trap in $G$ and $A$ is a $\sigma$ trap in $G$, we have $X \cap A$ is a $\overline{\sigma}$ trap in $G[A]$. Furthermore, $X \cap A$ is a $\sigma$ trap in $G[X]$.

If $X \cap A$ is in $R_{\overline{\sigma}}$ then $X \cap A$ is a valid move for $\overline{\sigma}$ in the trap depth game on $G$ after $\sigma$ plays $X$. Therefore, $\sigma$ must have a response $Y$ that leads to a win in at most $k-1$ rounds; then $Y$ is a $\overline{\sigma}$ trap in $X \cap A$ and therefore also in $A$, so it is a valid first move for $\sigma$ in $G[A]$.

Otherwise, $X \cap A$ is in $R_\sigma$. We will make $X \cap A$ player $\sigma$'s first move in the trap depth game on $G[A]$. Assume $\overline{\sigma}$ has a response $X'$ so that $\sigma$ cannot win from $X'$ in at most $k-1$ rounds. Then $X'$ is a $\sigma$ trap in $G[X \cap A]$ and therefore also in $G[X]$, so $X'$ is a valid response for $\overline{\sigma}$ in the trap depth game on $G$ to the play $X$, contradicting the assumption.
\end{proof}

Finally, in preparation for the next section, we translate the above into the language of parity games.

Define the ``max'' of a set of vertices to be those vertices in the set with maximum priority. Recall that ${Traps}_{\sigma}(G)$ 
is the set of nonempty $\sigma$-traps in $G$. Then the condition $X \in R_\sigma$ becomes $\max(X) \subseteq V_\sigma$. For example, we may rewrite the statements $\Delta$:

\noindent $\Delta_{\sigma}(G, 0) :=$ FALSE ;\\
$\Delta_{\sigma}(G, k+1) := [\exists X \in {Traps}_{\overline\sigma}(G) \text{ such that } \max (X) \subseteq V_{\sigma}] \text{ and }$ \\
\indent{\indent{\indent{$[\forall Y \in {Traps}_{\sigma}(G[X]) \text{ we have } 
( \Delta_{\sigma}(G[Y], k) \text{ or } \max(Y) \subseteq V_{\sigma})]$.}}}

In the definition of trap-depth game, for example, when it is player $\sigma$'s turn, player $\sigma$ will choose a $\overline\sigma$ trap whose largest priority is of parity $\sigma$.

Recall for a parity game $G$ that $\text{TDA}(G,\sigma,k)$ returns the largest set $X$ which, as a first move for $\sigma$, allows $\sigma$ to win in at most $k$ rounds in the trap-depth game on $G$. Theorem \ref{intersectTrap} tells us that the $k$th trap-depth algorithm is robust in the following sense: if one determines that some vertices are winning for $\sigma$ and removes their attractor from the graph, either one removes all of $\text{TDA}(G,\sigma,k)$ or else one can find the rest of $\text{TDA}(G,\sigma,k)$ by repeatedly running the $k$th trap depth algorithm on the remaining set.

\section{Trap-Depth Algorithms for Parity Games} \label{parityalg}

In this section of the paper, all discussions are with regards to parity games. We present a collection of algorithms that return subsets of the vertices of a parity game, culminating in the Trap-Depth Algorithm (TDA). We will have two versions of TDA (which take different inputs). We will discuss the first of the algorithms later. The characterization of the second algorithm, TDA$(G,\sigma,k)$, follows easily from the first, and it takes as its inputs a parity game $G$, a player $\sigma$, and an integer $k$. We will ultimately show the following characterization of the second TDA algorithm:

\begin{theorem} \label{theorem_TDA}
TDA$(G,\sigma,k)$ returns the largest (possibly empty) set $S$ so that if $\sigma$ uses $S$ as a first move, $\sigma$ can guarantee a win in at most $k$ moves in the trap-depth game on $G$.
\end{theorem}

Note that, by Theorem \ref{main_theorem}, this implies in particular that TDA$(G,\sigma,k) \subseteq W_\sigma$.

\subsection{B\"uchi Games}

Due to the complexity of the TDA, we will first introduce some simpler algorithms. Understanding the simpler algorithms will help significantly in understanding the TDA. The overall structure of the TDA resembles that of a classical algorithm for solving B\"uchi games, which we present here. A B\"uchi game is a parity game in which all of the priorities are either $0$ or $1$\footnote{This is actually a simpler problem than B\"uchi games. The way we defined parity games, the player that chooses where to move the token from a given vertex $v$ is just based on the parity of $p(v)$. If we had given an alternative, equivalent, definition of parity games in such a way that these two notions were separated, i.e. if the player who chooses where to move the token may depend on $v$ itself, then parity games in which all priorities are either 0 or 1 would be B\"uchi games. Simpler algorithms exist under our simplified definition of B\"uchi games. However, the algorithm for solving the original problem is instructive, even when applied to the simplified B\"uchi games, so we present it here.}. Therefore, Odd wins a B\"uchi game if and only if Odd has a strategy that reaches vertices of priority $1$ infinitely many times. The algorithm takes as input a B\"uchi game $G$ and returns the winning region for Odd.

\begin{algorithm}
\caption{B\"uchi$(G=(V,E,p))$}
\begin{algorithmic}[1]

\STATE$\sigma:=\mbox{Odd}$
\STATE$T_{\text{prev}}:=\emptyset$
\STATE$T_{\text{cur}}:=V_{\sigma}$
\WHILE{$T_{\text{cur}} \neq T_{\text{prev}}$}
\STATE$T_{\text{prev}}:=T_{\text{cur}}$
\STATE$C:=\text{Attr}(G,T_{\text{prev}},\sigma)$
\STATE$T_{\text{cur}}:=T_{\text{prev}}\setminus\{v \in V_\sigma: N(v) \cap C=\emptyset\}$
\ENDWHILE
\RETURN $C$

\end{algorithmic}
\label{buchi}
\end{algorithm}

The classical algorithm for B\"uchi games begins each iteration of its while-loop with a set of target vertices $T_{\text{prev}} \subseteq V_{\mbox{Odd}}$. It then computes the attractor of $T_{\text{prev}}$. This provides the largest set of vertices from which Odd has a strategy for reaching the set $T_{\text{prev}}$. The attractor by itself, however, does not provide any strategy for $\sigma$ after the token reaches $T_{\text{prev}}$. To remedy this, the algorithm then tests each vertex in $T_{\text{prev}}$ to check if it has the option to continue this strategy by returning the token back to the attractor of $T_{\text{prev}}$. If not, then that vertex is removed from being a target. This process repeats until the set $T_{\text{prev}}$ stabilizes. 

We will first observe that the algorithm outputs a subset of Odd's winning region. To see this, take $C$ to be the output of the algorithm and $T$ to be the final set of target vertices (so that $C=\text{Attr}(G,T,\mbox{Odd})$). Consider the following strategy for Odd: from any vertex of $T$, Odd chooses to enter $C$ (this is possible by the termination condition of the algorithm). From any vertex of $C \setminus T$, Odd follows a strategy to reach $T$. This guarantees that the vertices of $T$ are visited infinitely often; since these vertices have priority $1$, this is a winning strategy for Odd.

Conversely, Even has a winning strategy from any vertex not in $C$. To see this, let $T_0,T_1,\ldots$ be the sequence of values of $T_{\text{cur}}$ at the beginning of the while loop in the execution of B\"uchi$(G)$ (with the final value repeated). Take $C_i = \mbox{Attr}(G,T_i,\mbox{Odd})$. Note that $T_i$, and therefore $C_i$, is a decreasing sequence. We claim that any odd vertex in $C_i$ must be in $T_i$. If some Odd vertex $w$ were in $C_i$ but not in $T_i$, then there must be some edge from $w$ into $C_i$. However, $w$ must have been removed from $T_j$ for some $j < i$, which means there is no edge from $w$ into $C_j$. However, the sequence of $C_j$ is decreasing, a contradiction.

We now proceed by induction to show that Even has a winning strategy from any vertex not in $C_i$. Note that $T_0 = V_{\mbox{Odd}}$. Therefore, any vertex not in $C_0$ is not in the attractor of $V_{\mbox{Odd}}$, so from any such vertex Even has a strategy that never visits any vertex of priority $1$. Take for an inductive hypothesis that, for some $i$, any vertex not in $C_i$ is winning for Even. Then $T_{i+1}=T_i \setminus \{v \in V_{\mbox{Odd}} : N(v) \cap C_i = \emptyset\}$. Assume for contradiction that there is some vertex $v$ that is not in $C_{i+1}$ so that $v$ is winning for Odd. Then $v$ must be in $C_i \setminus C_{i+1}$, as otherwise $v$ is winning for Even by the inductive hypothesis. Since $v$ is not in $C_{i+1}$, Even has a strategy starting from $v$ that avoids $T_{i+1}$ indefinitely. Then consider Even playing the following strategy: as long as the token remains in $C_i$, Even plays any strategy that avoids entering $T_{i+1}$. If the token ever leaves $C_i$, then Even has a winning strategy and uses it. Since $v$ is winning for Odd, Odd must have some winning strategy. Consider the play induced by Even playing the aforementioned strategy and Odd playing a winning strategy. This play cannot leave $C_i$, as otherwise Even wins. However, some vertex $w \in C_i$ of priority $1$ must be visited. Since Even avoids $T_{i+1}$ indefinitely, we must have that $w$ is not in $T_{i+1}$. Since all Odd vertices in $C_i$ are in $T_i$, we therefore have that $w$ is in $T_i \setminus T_{i+1}$. However, by definition of $T_{i+1}$, this means that there are no edges from $w$ into $C_i$, so the next vertex in the play is not in $C_i$, a contradiction. This completes the proof of the correctness of the classical B\"uchi games algorithm.

\subsection{$k=1$}

The algorithm for solving trap-depth $1$ parity games closely resembles that for B\"uchi games. The main difference is that, rather than taking the attractor of the target set, we take a ``safe" version of the attractor. This takes a parameter $\lambda$; the $\lambda$-safe attractor in $G$ of a set $X$ for player $\sigma$ is the set of vertices from which $\sigma$ has a strategy that guarantees $X$ will be reached and that, in the process, no vertices (excluding those in $X$) of priority at least $\lambda$ are visited.

\begin{algorithm} \caption{SafeAttr$(G=(V,E,p),\lambda,X,\sigma)$}
\begin{algorithmic}[1]

\STATE $C_{\text{prev}}:= \emptyset$
\STATE $C_{\text{cur}} := X$
\WHILE{$C_{\text{cur}} \neq C_{\text{prev}}$}
\STATE $C_{\text{prev}}:=C_{\text{cur}}$
\STATE $C_{\text{cur}}:=C_{\text{prev}} \cup \{v \in V_\sigma : p(v)<\lambda \wedge N(v) \cap C_{\text{prev}} \neq \emptyset\} \newline 
\cup \{v \in V_{\overline\sigma} : p(v)<\lambda \wedge N(v) \subseteq C_{\text{prev}}\}$
\ENDWHILE
\RETURN $C_{\text{cur}}$
\end{algorithmic}
\label{alg2}
\end{algorithm}

At each iteration of the ``while'' loop, the set $C_{\text{cur}}$ (initially $X$) is enlarged by adding any vertices (of priority less than $\lambda$) in $V_{\sigma}$ or in $V_{\overline\sigma}$ that, respectively, have an edge going into $C_{\text{cur}}$ or have only edges going into $C_{\text{cur}}$. Note the similarities to the attractor, Algorithm~\ref{alg1}.

Indeed, one sees that $\text{SafeAttr}(G,\lambda,X,\sigma) = \text{Attr}(G,X,\sigma)$ if $\lambda \geq p(\text{max}(V \setminus X))$. And, just like the regular Attractor, one sees that the Safe Attractor stabilizes its own output; i.e., $\text{SafeAttr}(G,\lambda,\text{SafeAttr}(G,\lambda,X,\sigma),\sigma) = \text{SafeAttr}(G,\lambda,X,\sigma)$.

However, it is not obvious how to directly substitute the safe attractor into Algorithm \ref{buchi}, as there does not appear to be a canonical choice for the parameter $\lambda$. This motivates the Sequential Safe Attractor algorithm, which, in the trap-depth $k=1$ case, iteratively applies the $\lambda$-safe attractor to $\sigma$-vertices of priority at least $\lambda$. Recall that we defined the max of a set of vertices to be the vertices of largest priority in that set. Below, if $S$ is a set of vertices that all have the same priority, then $p(S)$ is that priority (rather than the singleton containing that priority).

\begin{algorithm}
\caption{SeqAttr$_1(G=(V,E,p),X,\sigma)$}
\begin{algorithmic}[1]
\STATE $W:=V_\sigma \cap X$
\STATE $C :=\emptyset$
\WHILE{$W \neq \emptyset$}
\STATE $S := \max(W)$
\STATE $C:= \text{SafeAttr}(G,p(S),C \cup S,\sigma)$
\STATE $W:=W \setminus C$
\ENDWHILE
\RETURN $C$
\end{algorithmic}
\end{algorithm}

At the beginning of the ``while'' loop above, we have a set $C$ and a list $W$ of target vertices to process. Each iteration of the loop calls the Safe Attractor Algorithm. The Sequential Attractor removes the issue of a priority bound inside the Safe Attractor. For any vertex $v \in V$, SeqAttr$_1(G,X,\sigma)$ tests if $\sigma$ has a strategy to move the token from $v$ towards some $w \in X$ in which any resulting path did not visit any vertices of priority at least $p(w)$.

The algorithm for solving trap-depth $1$ parity games, TDA$_1(G,\sigma)$, simply substitutes the Sequential Safe Attractor for the Attractor in the B\"uchi games algorithm, Algorithm \ref{buchi}.

\begin{algorithm}
\caption{TDA$_1(G=(V,E,p),\sigma)$}
\begin{algorithmic}[1]

\STATE$T_{\text{prev}}:=\emptyset$
\STATE$T_{\text{cur}}:=V_{\sigma}$
\WHILE{$T_{\text{cur}} \neq T_{\text{prev}}$}
\STATE$T_{\text{prev}}:=T_{\text{cur}}$
\STATE$C:=\text{SeqAttr}_1(G,T_{\text{prev}},\sigma)$
\STATE$T_{\text{cur}}:=T_{\text{prev}}\setminus\{v \in V_\sigma: N(v) \cap C=\emptyset\}$
\ENDWHILE
\RETURN $C$

\end{algorithmic}
\end{algorithm}

TDA$_1(G,\sigma)$ returns the largest $\overline{\sigma}$ trap $X$ in $G$ so that every $\sigma$-subtrap $Y$ has that the vertices of largest priority in $Y$ belong to player $\sigma$. We will sketch a proof of this fact; a complete proof will follow from the arguments in the next section. The proof will argue three different points, from which the characterization of TDA$_1(G,\sigma)$ follows immediately:

\begin{enumerate}
\item (Monotonicity): If $A$ is a $\overline\sigma$-trap in $G$, then we have TDA$_1(G[A],\sigma) \subseteq \text{TDA}_1(G,\sigma)$.

\item (Completeness): If $V$ satisfies that every $\sigma$-trap in $V$ has a vertex whose maximum priority is of parity $\sigma$, then TDA$_1(G=(V,E,p),\sigma) = V$.

\item (Soundness): TDA$(G,\sigma)$ is a $\overline\sigma$-trap in $G$ whose maximum priority is of parity $\sigma$ and which satisfies that every $\sigma$-subtrap has maximum priority $\sigma$.
\end{enumerate}

To argue monotonicity, one first argues that the Safe Attractor Algorithm is monotonic with respect to its parameter $\lambda$, its input set $X$, and sometimes with respect to the parity game. 

Explicitly, if $\lambda_1 \leq \lambda_2$, if $X_1 \subseteq X_2$, and if $A$ is a $\overline\sigma$-trap in $G$, then 
$$\text{SafeAttr}(G[A],\lambda_1,X_1,\sigma) \subseteq \text{SafeAttr}(G,\lambda_2,X_2,\sigma).$$ This is intuitive, but will be proven carefully in the next section.

From this, monotonicity of SeqAttr$_1$ easily follows. If $X_1 \subseteq X_2$ and if $A$ is a $\overline\sigma$-trap in $G$, then SeqAttr$_1(G[A],X_1,\sigma) \subseteq \text{SeqAttr}_1(G,X_2,\sigma)$. A generalization of this will be proven in the next section.

Finally, from this, point (1), monotonicity of TDA$_1$, easily follows. Again, a generalization is carefully proved in the next section.

Now we will argue completeness. Assume the whole vertex set $V$ satisfies that every $\sigma$-trap in $V$ has maximum priority of parity $\sigma$. Then we claim that SeqAttr$_1(G,V,\sigma) = V$. This will imply that TDA$_1$ terminates after the first call to SeqAttr and returns $V$, as desired. To see that SeqAttr$_1(G,V,\sigma)=V$, consider the first iteration of the ``WHILE" loop: since $V$ is a $\sigma$-trap in $V$, the first iteration computes the $\sigma$-safe-attractor of $\max(V)$ with $\lambda=p(\max(V))$. However, since these are the vertices of maximum priority, this is the same as computing the attractor. The remaining set is a $\sigma$-trap in $V$, and so has maximum priority of parity $\sigma$. By induction, this continues until $W$ is empty and SeqAttr$_1$ returns all of $V$.

Finally, we argue the soundness of TDA$_1$. At the end of the execution of TDA$_1$, there is some final set of target vertices $T$ satisfying every vertex of $T$ has an edge into SeqAttr$_1(G,T,\sigma)$. By our observations about SeqAttr$_1$, starting from any vertex in TDA$_1(G,\sigma)$, $\sigma$ has some strategy to reach some vertex $w$ in $T$ so that along the way only priorities less than $p(w)$ are visited. Furthermore, by the terminating condition of TDA$_1$, every vertex in $T$ has an edge back into TDA$_1(G,\sigma)$. This gives that TDA$_1(G,\sigma)$ is indeed a $\overline\sigma$-trap. Furthermore, starting from the largest vertex in any $\sigma$-trap, $\sigma$ may follow the strategy that allows $\sigma$ to reach some vertex $w$ in $T$ without seeing larger vertices along the way; $\sigma$ cannot force the play to leave the trap, and therefore the largest priority in the trap must be $w$, a $\sigma$-vertex, as desired.

\subsection{General Trap-Depth Algorithm}

The main difference between the general Trap-Depth Algorithm (TDA) and TDA$_1$ is that we change the call to Safe Attractor in TDA$_1$ to a stronger algorithm, the Generalized Safe Attractor. Indeed, we prove a more general result than Theorem \ref{theorem_TDA}, allowing us to strengthen any algorithm that satisfies certain conditions.

\begin{definition}
Let ParAlg$(G,\sigma)$ be an algorithm that takes as input a parity game $G$ and a player $\sigma$ and returns a subset of the vertices of $G$. We say ParAlg is \emph{nice with traps} if:

\begin{itemize}
\item For any parity game $G$ and any $\overline{\sigma}$ trap $A$ in $G$, ParAlg$(G[A],\sigma)\subseteq \text{ ParAlg}(G,\sigma)$.
\item For any parity game $G$, taking $S = \text{ParAlg}(G,\sigma)$, for any $\overline{\sigma}$-trap $A$ containing $S$, ParAlg$(G[A],\sigma)=S$, and for any $\sigma$-trap $A$ intersecting $S$, ParAlg$(G[A],\sigma)$ is nonempty.
\item ParAlg$(G,\sigma)$ always returns a $\overline{\sigma}$-trap whose largest priority belongs to player $\sigma$.
\end{itemize}
\end{definition}

If ParAlg is nice with traps, then we can strengthen it via the TDA.

\begin{theorem}\label{GeneralTDA}
Let ParAlg$(G,\sigma)$ be an algorithm that takes as input a parity game $G$ and a player $\sigma$ and returns a subset of the vertices of $G$. If ParAlg is nice with traps, then TDA$(G,\sigma,\text{ParAlg})$ returns the largest $\overline{\sigma}$-trap $X$ whose maximum priority is of parity $\sigma$ so that, for every $\sigma$-subtrap $Y$, either the maximum priority of $Y$ belongs to $\sigma$ or we have ParAlg$(G[Y],\sigma)$ is nonempty.
\end{theorem}

Recursively applying Theorem \ref{GeneralTDA} will give Theorem \ref{theorem_TDA}.

We will introduce terminology to talk about the conditions on the set $X$ from Theorem \ref{GeneralTDA}.

\begin{definition}
Let ParAlg$(G,\sigma)$ be an algorithm that takes as input a parity game $G$ and a player $\sigma$ and returns a subset of the vertices of $G$. Given a parity game $G$ and a set of vertices $X$, we say $X$ is $\emph{good}$ for ParAlg with respect to player $\sigma$ if $X$ is a $\overline{\sigma}$-trap whose maximum priority is of parity $\sigma$ so that for every $\sigma$-subtrap $Y$, either the maximum priority of $Y$ belongs to $\sigma$ or we have ParAlg$(G[Y],\sigma)$ is nonempty.
\end{definition}

When the context is clear, we will simply say that $X$ is good.

We said that TDA is obtained from TDA$_1$ by replacing the call to safe attractor, so let us first present the Generalized Safe Attractor Algorithm. The idea behind the safe attractor is to guarantee reaching a target set of vertices in a $\lambda$-safe way, so that along the way we don't see vertices of priority at least $\lambda$. The idea behind the generalized safe attractor is to guarantee that, if $\sigma$ fails to reach a target set of vertices, then $\sigma$ wins the play; this all happens in a manner that is careful with vertices of priority at least $\lambda$, which we informally refer to as ``$\lambda$-safe".

In order to ensure that everything that is done is $\lambda$-safe, we will at some point remove all vertices of priority at least $\lambda$ from the game by taking the restriction to vertices of lower priority: define the $\lambda$-restriction of a parity game Restrict$(G=(V,E,p),\lambda,\sigma):=V \setminus \text{Attr} (G,\left\{ v \in V : p(v) \geq \lambda \right\},\overline\sigma )$. In words, the only vertices that remain in Restrict$(G,\lambda,\sigma)$ are those from which $\sigma$ can ensure that all the priorities in any resulting play are less than $\lambda$.

We are now ready to introduce the Generalized Safe Attractor. It takes as input a parity game $G$, a number $\lambda$, a set of target vertices $X$, a player $\sigma$, and an algorithm ParAlg$(G,\sigma)$. It is most useful to think of the context where ParAlg is nice with traps.

\begin{algorithm}

\caption{GenAttr$(G=(V,E,p),\lambda,X,\sigma,\text{ParAlg})$}

\begin{algorithmic}[1]

\STATE $C_{\text{prev}}:= \emptyset$
\STATE $C_{\text{cur}} := X$
\WHILE{$C_{\text{cur}} \neq C_{\text{prev}}$}
\STATE $C_{\text{prev}}:=C_{\text{cur}}$
\STATE $S:= \text{SafeAttr}(G,\lambda,C_{\text{prev}},\sigma)$
\STATE $V' := \text{Restrict}(G[V \setminus S],\lambda,\sigma)$
\STATE $C_{\text{cur}}:=S \cup \text{ParAlg}(G[V'],\sigma)$
\ENDWHILE
\RETURN $C_{\text{cur}}$

\end{algorithmic}

\end{algorithm}

At the general step of the ``while'' loop, we begin with a set $C_{\text{prev}}$ of vertices which we want to reach in a $\lambda$-safe way. The loop then calls the Safe Attractor Algorithm. SafeAttr$(G,\lambda,C_{\text{prev}},\sigma)$ returns the largest collection of vertices $S$ from which $\sigma$ has a strategy to force the token into $C_{\text{prev}}$ such that the token only hits vertices of priority smaller than $\lambda$ along the way. Once the set $S$ has been found, we check if, given that $\overline{\sigma}$ avoids $X$, player $\sigma$ has any \emph{winning} strategy given by ParAlg (which is $\lambda$-safe) on the remaining set $V \setminus S$; we add this to $S$ to get $C_{\text{cur}}$. Each iteration either adds vertices to $C_{\text{cur}}$ or terminates the loop. Since $\left|C_{\text{cur}}\right|$ cannot increase indefinitely, GenAttr eventually halts.

For a vertex $v$ and a subset $X$, $v$ will be in GenAttr if and only if there is a $\sigma$-strategy to move the token from $v$ towards $X$ such that, depending on $\overline\sigma$'s moves, \emph{either} the token eventually reaches $X$ \emph{or} the token reaches a vertex from which ParAlg provides a winning strategy for $\sigma$; in both cases, all the vertices visited by the token have priority less than $\lambda$, except perhaps the ones in $X$.

As the name suggests, the Generalized Safe Attractor is a generalization of the Safe Attractor. Consider the case where ParAlg$(G,\sigma)$ simply returns the empty set for every input. Under this condition, we claim that GenAttr$(G,\lambda,X,\sigma,\text{ParAlg}) = \text{SafeAttr}(G,\lambda,X,\sigma)$. Later we will show that SafeAttr stabilizes its own output; this immediately gives that, if ParAlg is always empty, a call to GenAttr will have at the end of the first ``WHILE" loop $C_{\text{cur}}$ equal to SafeAttr, and will subsequently terminate with this output.

Both the general case of the sequential safe attractor algorithm and the TDA are analogous to the ones before, except the sequential safe attractor calls the generalized safe attractor. As before, if $S$ is a set of vertices that all have the same priority, then we write $p(S)$ to denote that priority.

\begin{algorithm}
\caption{SeqAttr$(G=(V,E,p),X,\sigma,\text{ParAlg})$}
\begin{algorithmic}[1]
\STATE $W:=V_\sigma \cap X$
\STATE $C :=\emptyset$
\WHILE{$W \neq \emptyset$}
\STATE $S := \max(W)$
\STATE $C:= \text{GenAttr}(G,p(S),C \cup S,\sigma,\text{ParAlg})$
\STATE $W:=W \setminus C$
\ENDWHILE
\RETURN $C$
\end{algorithmic}
\label{alg3}
\end{algorithm}

For any $v \in \text{SeqAttr}(G,X,\sigma,\text{ParAlg})$, $\sigma$ has a strategy to ensure that, if the token ever reaches some $w \in X$, it hits only vertices of priority smaller than $p(w)$ along the way and, if it never reaches $X$, then $\sigma$ wins.

TDA simply calls the sequential safe attractor.

\begin{algorithm}
\caption{TDA$(G=(V,E,p),\sigma,\text{ParAlg})$}
\begin{algorithmic}[1]
\STATE$T_{\text{prev}}:=\emptyset$
\STATE$T_{\text{cur}}:=V_{\sigma}$
\WHILE{$T_{\text{cur}} \neq T_{\text{prev}}$}
\STATE$T_{\text{prev}}:=T_{\text{cur}}$
\STATE$C:=\text{SeqAttr}(G,T_{\text{prev}},\sigma,\text{ParAlg})$
\STATE$T_{\text{cur}}:=T_{\text{prev}}\setminus\{v \in V_\sigma: N(v) \cap C=\emptyset\}$
\ENDWHILE
\RETURN $C$

\end{algorithmic}
\label{alg4}
\end{algorithm}

As before, TDA calls SeqAttr on progressively smaller sets of target vertices $T_{\text{cur}}$.

One easily sees that the set $C$ output by TDA$(G,\sigma,\text{ParAlg})$ is the Sequential Attractor of some final collection $T \subseteq V_{\sigma}$ of target vertices. If $\sigma$ follows the strategy given by the sequential attractor on $C$, the resulting play will either be winning for $\sigma$, as guaranteed by the conditions on ParAlg, or reach $T$ infinitely often, and the largest priority will belong to $\sigma$ and so the play will be winning for $\sigma$.

\subsection{Correctness}

We now outline the proof of Theorem \ref{GeneralTDA}. We will prove three propositions from which Theorem \ref{GeneralTDA} will follow immediately. Note that, in the second statement below (completeness), $V$ is the whole vertex set of the parity game. 

\begin{theorem}
Let ParAlg$(G,\sigma)$ be an algorithm that takes as input a parity game $G$ and a player $\sigma$ and returns a subset of the vertices of $G$. Assume ParAlg is nice with traps.

\begin{enumerate}

\item (Monotonicity): If $A$ is a $\overline\sigma$-trap in $G$, then we have TDA$(G[A],\sigma,\text{ParAlg}) \subseteq \text{TDA}(G,\sigma,\text{ParAlg})$.

\item (Completeness): If $V$ is good, then TDA$(G,\sigma,\text{ParAlg}) = V$.

\item (Soundness): TDA$(G,\sigma,\text{ParAlg})$ is good.

\end{enumerate}
\end{theorem}

We will now build up the machinery to prove the above.

\subsubsection{General Lemmas}

We begin with some general lemmas that will be useful; they are all reasonably simple to prove and we omit their proofs.

The first of these notes that SafeAttr is equal to Attr for $\lambda$ large enough.

\begin{lemma}
$\text{SafeAttr}(G,\lambda,X,\sigma)=\text{Attr}(G,X,\sigma)$ if $\lambda > p(\max(V \setminus X))$.
\end{lemma}

The next lemma notes that the algorithms are stable when run on their own outputs. The second and third statements follow from the ones prior.

\begin{lemma} Let $A$ be a $\overline{\sigma}$ trap.
\begin{enumerate}
	\item {If $S := \text{SafeAttr}(G,\lambda,X,\sigma) \subseteq A$, then
$$S=\text{SafeAttr}(G,\lambda,S,\sigma)=\text{SafeAttr}(G[A],\lambda,S,\sigma)$$}
	\item{If $S := \text{GenAttr}(G,\lambda,X,\sigma,\text{ParAlg}) \subseteq A$, then
$$S=\text{GenAttr}(G,\lambda,S,\sigma,\text{ParAlg})=\text{GenAttr}(G[A],\lambda,S,\sigma,\text{ParAlg})$$}
	\item{If $S := \text{SeqAttr}(G,\lambda,X,\sigma,\text{ParAlg}) \subseteq A$, then
$$S=\text{SeqAttr}(G,\lambda,S,\sigma,\text{ParAlg})=\text{SeqAttr}(G[A],\lambda,S,\sigma,\text{ParAlg})$$}
\end{enumerate}
\end{lemma}

A similar statement holds for the TDA. Observe first that TDA$(G,\sigma,k)$ is a $\overline{\sigma}$ trap in $G$ with maximum vertex of parity $\sigma$.

\begin{lemma}
Take $S:= \text{TDA}(G,\sigma,\text{ParAlg})$. Then
$$S=\text{TDA}(G[S],\sigma,\text{ParAlg}).$$
\end{lemma}

As before, it is also true that $\text{TDA}(G,\sigma,\text{ParAlg}) = \text{TDA}(G[A],\sigma,\text{ParAlg})$ where $A$ is any $\overline{\sigma}$ trap containing $\text{TDA}(G,\sigma,\text{ParAlg})$; this follows easily from the characterization of the TDA, but is more difficult to prove given only what we have established so far.

\subsubsection{Monotonicity}

We now prove monotonicity of SafeAttr for the inputs $X,\lambda$ and sometimes for the graph $G$:

\begin{lemma}
If $X_1 \subseteq X_2$, $\lambda_1 \leq \lambda_2$, and $A$ is a $\overline{\sigma}$ trap with $X_1 \subseteq A$, then SafeAttr$(G[A],\lambda_1,X_1,\sigma) \subseteq \text{SafeAttr}(G,\lambda_2,X_2,\sigma)$.
\end{lemma}

\begin{proof}
Take $C_0,C_1,\ldots$ to be the values of $C_{\text{cur}}$ at the beginning of each ``while'' loop in the execution of SafeAttr$(G[A],\lambda_1,X_1,\sigma)$ (with the final value repeating) and take $C_0',C_1',\ldots$ similarly from the execution of SafeAttr$(G,\lambda_2,X_2,\sigma)$. Then $C_0\subseteq C_0'$. 

Note that, by definition of a trap, if $v \in A \cap V_\sigma$ then $N_{G[A]}(v) \subseteq N_G(v)$, and if $v \in A \cap V_{\overline{\sigma}}$ then $N_{G[A]}(v) = N_G(v)$.

If $C_i \subseteq C_i'$ then we have

$$\{v \in A \cap V_\sigma : p(v)<\lambda_1 \wedge N_{G[A]}(v) \cap C_i \neq \emptyset\} \subseteq \{v \in V_\sigma : p(v) < \lambda_2 \wedge N_G(v) \cap C_i' \neq \emptyset\}$$
$$\{v \in A \cap V_{\overline\sigma} : p(v)<\lambda_1 \wedge N_{G[A]}(v) \subseteq C_{i}\} \subseteq \{v \in V_{\overline\sigma} : p(v)<\lambda_2 \wedge N_G(v) \subseteq C_i'\}$$

and so $C_{i+1} \subseteq C_{i+1}'$, and by induction this holds for all $i$.
\end{proof}

We now simultaneously address three monotonicity properties for GenAttr: monotonicity of the output with respect to the inputs $X,\lambda$ and also sometimes with respect to the graph $G$.

The next lemma is a weak monotonicity property for GenAttr, saying that one iteration of the `WHILE' loop in the algorithm will be contained in the GenAttr of the stronger inputs; combining this with the previous lemma will give monotonicity properties for GenAttr.

\begin{lemma}
Assume $X_1 \subseteq X_2$ and $\lambda_1 \leq \lambda_2$. Assume $A$ is a $\overline{\sigma}$ trap in $G$ and $X_1 \subseteq A$. Take
$$S^*:=\text{SafeAttr}(G[A],\lambda_1,X_1,\sigma)$$
$$V^*:=\text{Restrict}(G[A\setminus S^*],\lambda_1,\sigma)$$
$$T^*:=\text{ParAlg}(G[V^*],\sigma)$$
Then $S^* \cup T^* \subseteq \text{GenAttr}(G,\lambda_2,X_2,\sigma,\text{ParAlg})$.
\end{lemma}

\begin{proof}
Take $C'$ to be the value of $C_{\text{cur}}$ at the end of the execution of GenAttr$(G,\lambda_2,X_2,\sigma,\text{ParAlg})$. Taking:

$$S':=\text{SafeAttr}(G,\lambda_2,C',\sigma)$$
$$V':=\text{Restrict}(G[V\setminus S'],\lambda_2,\sigma)$$
$$T':=\text{ParAlg}(G[V'],\sigma)$$
we have that $C'=S'$ and $T'$ is empty.

We get that $S^* \subseteq S' = C'$ by monotonicity of the safe attractor, and so we need only show that $T^* \subseteq C'$. Assume, for sake of contradiction, that $T^*$ is not a subset of $C'$, and so in particular $T^* \setminus S' \neq \emptyset$.

In the following, we refer to traps in induced subgraphs that are not necessarily subgames, i.e. they may have vertices without outgoing edges. The definition of trap remains unchanged.

We claim that $T^* \setminus S'$ is a $\overline{\sigma}$ trap in $G[V']$. 
Take $B := \{v \in A : p(v) \geq \lambda_1\}$. We know that $T^*$ is a $\overline{\sigma}$ trap in $G[V^*]$. Then note that $V^* = (A \setminus S^*) \setminus \text{Attr}(G[A \setminus S^*],B,\overline{\sigma})$, and so we get that $V^*$ is a $\overline{\sigma}$ trap in $G[A \setminus S^*]$ and so $T^*$ is a $\overline{\sigma}$ trap in $G[A \setminus S^*]$. Since the edges of $G[A \setminus S']$ are a subset of the edges of $G[A \setminus S^*]$, we get that any $\overline{\sigma}$ vertex in $T^* \setminus S'$ has no edges leaving $T^* \setminus S'$ in the graph $G[A \setminus S']$. Given any $\sigma$ vertex in $T^* \setminus S'$, since $S' = \text{SafeAttr}(G,\lambda_2,S',\sigma)$ and since $\forall v \in V^* \, p(v) < \lambda_1 \leq \lambda_2$, the $\sigma$ vertex had no edges into $S'$ in $G[A]$ (for otherwise it would be contained in $S'$) and so the $\sigma$ vertex must have some edge into $T^* \setminus S'$ and so we get that indeed $T^* \setminus S'$ is a $\overline{\sigma}$ trap in $G[A \setminus S']$, and so also in $G[V \setminus S']$ (since $A$ is a $\overline{\sigma}$-trap in $V$). We have $T^* \setminus S'$ is a $\overline\sigma$-trap in $G[V \setminus S']$ with no vertices of priority at least $\lambda$; any such structure must be a $\overline\sigma$-trap in $G[V']$.

Because $T^*$ has no vertices of priority at least $\lambda_2$, we have $T^* \setminus S' = T^* \setminus \text{Attr}(G[T^*],S',\sigma)$, and so $T^* \setminus S'$ is a $\sigma$-trap in $G[T^*]$. Since we know ParAlg is nice with traps, we have that ParAlg$(G[T^*],\sigma)=T^*$. Therefore, again since ParAlg is nice with traps, ParAlg$(G[T^* \setminus S'],\sigma)$ is nonempty. Finally, ParAlg$(G[T^* \setminus S'],\sigma) \subseteq \text{ParAlg}(G[V'],\sigma)$, but this contradicts the assumption that $T' = \emptyset$. Therefore, we must have that $T^* \subseteq S'$.
\end{proof}

\begin{lemma}
If $X_1 \subseteq X_2$ and $\lambda_1 \leq \lambda_2$ and $A$ is a $\overline{\sigma}$ trap in $G$ with $X_1 \subseteq A$, then GenAttr$(G[A],\lambda_1,X_1,\sigma,\text{ParAlg}) \subseteq \text{GenAttr}(G,\lambda_2,X_2,\sigma,\text{ParAlg})$.
\end{lemma}

\begin{proof}
Take $C_0,C_1,\ldots$ to be the values of $C_{\text{cur}}$ at the beginning of each ``while'' loop in the execution of GenAttr$(G[A],\lambda_1,X_1,\sigma,\text{ParAlg})$ (with the final value repeating). Take

$$S_i:=\text{SafeAttr}(G[A],\lambda_2,C_i,\sigma)$$
$$V_i:=\text{Restrict}(G[A\setminus S_i],\lambda_2,\sigma)$$
$$T_i:=\text{ParAlg}(G[V_i],\sigma)$$

We will proceed by induction on $i$ to show that $C_i \subseteq \text{GenAttr}(G,n_2,X_2,\sigma,\text{ParAlg})$. Note this holds for $C_0$ since $C_0 = X_1 \subseteq X_2$. Then, for $i > 0$, we have $C_i = T_{i-1} \cup S_{i-1}$ and by the previous lemma and inductive hypothesis we get:
$$T_{i-1} \cup S_{i-1} \subseteq \text{GenAttr}(G,\lambda_2,C_{i-1},\sigma,\text{ParAlg}) \subseteq $$
$$\text{GenAttr}(G,\lambda_2,\text{GenAttr}(G,\lambda_2,X_2,\sigma,\text{ParAlg}),\sigma,\text{ParAlg}) =$$
$$\text{GenAttr}(G,\lambda_2,X_2,\sigma,\text{ParAlg}),$$ completing the proof.
\end{proof}

We will now proceed to show monotonicity of SeqAttr. While the original definition was slightly more natural, the following reformulation of SeqAttr will be more useful. We leave it to the reader to verify that the following reformulation of SeqAttr is equivalent to the original. It follows immediately from monotonicity and stability of GenAttr.

\begin{lemma}
If $P'$ is a finite collection of integers and $p(X \cap V_\sigma) \subseteq P'$ then, taking $P$ to be the priorities in $P'$ of parity $\sigma$, the following algorithm has the same output as SeqAttr.
\begin{algorithm}

\caption{SeqAttr$_P(G=(V,E,p),X,\sigma,\text{ParAlg})$}

\begin{algorithmic}[1]
\STATE $C :=\emptyset$
\STATE $Q :=P$
\WHILE{$Q \neq \emptyset$}
\STATE $\lambda := \max(Q)$
\STATE $Q := Q \setminus \{\lambda\}$
\STATE $S := \{v \in X : p(v) = \lambda\}$
\STATE $C:= \text{GenAttr}(G,\lambda,C \cup S,\sigma,\text{ParAlg})$
\ENDWHILE
\RETURN $C$

\end{algorithmic}
\label{alg3rev}
\end{algorithm}
\end{lemma}

Intuitively, we simply run the GenAttr for every priority in $P$, which just adds redundancy by the assumption $p(X\cap V_\sigma) \subseteq P$: if ever in the original formulation of SeqAttr some call GenAttr$(G,\lambda,C,\sigma,\text{ParAlg})$ were made, then in the above version some call will be made with the same parameter $\lambda$.

We now show monotonicity properties for SeqAttr with respect to the input $X$ and also sometimes with respect to the graph $G$:

\begin{lemma}
If $X_1 \subseteq X_2$ and $A$ is a $\overline \sigma$-trap in $G$ such that $X_1 \subseteq A$, then SeqAttr$(G[A],X_1,\sigma,\text{ParAlg}) \subseteq \text{SeqAttr}(G,X_2,\sigma,\text{ParAlg})$
\end{lemma}
\begin{proof}
Take $P = p(X_2 \cap V_\sigma)$. Take $C_i,Q_i$ to be the values of $C,Q$ respectively at the beginning of the $i$th iteration of the ``WHILE'' loop in the execution of SeqAttr$_P(G[A],X_1,\sigma,\text{ParAlg})$. Take

$$\lambda_i = \max(Q_i)$$
$$S_i = \{v \in X_1 : p(v) = \lambda_i\}$$

Similarly take $C_i',Q_i',\lambda_i',S_i'$ for the execution of SeqAttr$_P(G,X_2,\sigma,\text{ParAlg})$. Since $Q_0 = Q_0'$ and $Q_{i+1} = Q_i \setminus \max(Q_i)$ and $Q_{i+1}' = Q_i' \setminus \max(Q_i')$ we get $Q_i = Q_i'$ and $\lambda_i = \lambda_i'$ for all $i$. Then $S_i \subseteq S_i'$ since $X_1 \subseteq X_2$. We now proceed by induction to show $C_i \subseteq C_i'$. We have $C_0 = C_0' = \emptyset$ and 
$$C_{i+1} = \text{GenAttr}(G[A],\lambda_i,S_i \cup C_i,\sigma,\text{ParAlg}) \subseteq$$ 
$$\text{GenAttr}(G,\lambda_i',S_i' \cup C_i',\sigma,\text{ParAlg}) = C_{i+1}'$$ 

\end{proof}

We now present the monotonicity theorem for TDA:

\begin{theorem}
If $A$ is a $\overline\sigma$-trap in $G$, then we have TDA$(G[A],\sigma,\text{ParAlg}) \subseteq \text{TDA}(G,\sigma,\text{ParAlg})$.
\end{theorem}
\begin{proof}
Let $T_0,T_1,\ldots$ be the values of $T_{\text{cur}}$ at the beginning of the ``WHILE'' loop in the execution of $\text{TDA}(G[A],\sigma,\text{ParAlg})$. Take $C_i := \text{SeqAttr}(G[A],T_i,\sigma,\text{ParAlg})$. Similarly define $T_i',C_i'$ for TDA$(G,\sigma,\text{ParAlg})$. We proceed by induction to show $T_i \subseteq T_i'$. This holds for $T_0,T_0'$ since $A \cap V_\sigma \subseteq V_\sigma$. Then, by monotonicity of SeqAttr, we get $C_i \subseteq C_i'$. To obtain $T_{i+1}$ and $T_{i+1}'$ from $T_i,T_i'$, respectively, any vertex in $T_i,T_i'$ without an edge into $C_i,C_i'$ is removed. The edges of $G$ are a superset of those of $G[A]$ and $C_i'$ is a superset of $C_i$, so we get:
$$T_{i+1}= T_i \setminus \{v \in V_\sigma \cap A : N_{G[A]}(v) \cap C_i = \emptyset\} = T_i \setminus \{v \in T_i: N_{G[A]}(v) \cap C_i = \emptyset\} \subseteq$$
$$T_i' \setminus \{v \in T_i' : N_G(v) \cap C_i' = \emptyset\} = T_i' \setminus \{v \in V_\sigma : N_G(v) \cap C_i' = \emptyset\} = T_{i+1}'.$$
\end{proof}

\subsubsection{Completeness}

\begin{lemma}
If $V$ is good for ParAlg with respect to $\sigma$ and if $T \subseteq V$ and $\lambda$ are such that $p(\max(V \setminus T)) < \lambda$, then taking $S:=\text{GenAttr}(G,T,\lambda,\sigma,\text{ParAlg})$ we have $S = \text{Attr}(G,S,\sigma)$ and if $V\setminus S \neq \emptyset$ then $\max(V \setminus S) \subseteq V_\sigma$.
\end{lemma}
\begin{proof}
We consider that by the terminating condition for the GenAttr algorithm, we must have
$$S = \text{SafeAttr}(G,S,\lambda,\sigma).$$
Note that $\text{SafeAttr}(G,S,\lambda,\sigma)=\text{Attr}(G,S,\sigma)$ (since $\lambda > p(\max(V \setminus T))$) and so we get
$$S = \text{Attr}(G,S,\sigma).$$
Since $p(\max(V \setminus S)) < \lambda$, we also get by the terminating condition for GenAttr that $\text{ParAlg}(V \setminus S,\sigma) = \emptyset$, but, by assumption, since $V \setminus S$ is a $\sigma$ trap in $G$, either $V\setminus S = \emptyset$ (in which case we are done) or $\max(V \setminus S) \subseteq V_\sigma$.
\end{proof}

\begin{lemma}
If $V$ is good for ParAlg with respect to $\sigma$, then SeqAttr$(G,V_\sigma,\sigma,\text{ParAlg}) = V$
\end{lemma}
\begin{proof}
Taking $W_0,W_1,\ldots$ to be the values of $W$ at the beginning of each while loop in the execution of SeqAttr$(G,V,\sigma,\text{ParAlg})$, take
$$S_i := \max(W_i)$$
$$C_i := \text{GenAttr}(G,p(S_i),C_{i-1} \cup S_i,\sigma,\text{ParAlg})$$
then we have by the previous lemma $\max(W_0) = \max(V)$. Then, by induction, if $\max(W_i) \in V_\sigma$, we have either $W_{i+1} = \emptyset$ or $\max(W_{i+1}) = \max(V \setminus C_i)$ and so the SeqAttr will not terminate until $V \subseteq C_i$.
\end{proof}

\begin{theorem}
If $V$ is good for ParAlg with respect to $\sigma$, then TDA$(G,\sigma,\text{ParAlg}) = V$.
\end{theorem}
\begin{proof}
The previous lemma immediately gives that the TDA will terminate after the first iteration of the ``WHILE'' loop and return $V$, since SeqAttr will return the whole set of vertices.
\end{proof}

\subsubsection{Soundness}
\begin{lemma}
If $X$ is a $\sigma$-trap in $G$ and if $X \cap Y = \emptyset$ then $X \cap \text{SafeAttr}(G,Y,\lambda,\sigma) = \emptyset$.
\end{lemma}
\begin{proof}
Note that SafeAttr$(G,Y,\lambda,\sigma) \subseteq \text{Attr}(G,Y,\sigma)$ and $\text{Attr}(G,Y,\sigma) \cap X = \emptyset$.
\end{proof}

\begin{lemma}
If $X$ is a $\sigma$-trap in $G$ with largest vertex of priority $m < \lambda$ and with ParAlg$(G[X],\sigma) = \emptyset$, then if $X \cap Y = \emptyset$ we have $X \cap \text{GenAttr}(G,\lambda,Y,\sigma,\text{ParAlg}) = \emptyset$.
\end{lemma}
\begin{proof}
Take $C_0,C_1,\ldots$ to be the value of $C_{\text{cur}}$ at the beginning of each ``WHILE'' loop in the execution of GenAttr$(G,\lambda,Y,\sigma,\text{ParAlg})$. Take
$$S_i:=\text{SafeAttr}(G,\lambda,C_i,\sigma)$$
$$V_i:=\text{Restrict}(G[V\setminus S_i],\lambda,\sigma)$$
$$T_i:=\text{ParAlg}(G[V_i],\sigma)$$
We proceed by induction on $i$ to show $C_i \cap X = \emptyset$. This holds by assumption for $C_0 = Y$. If this holds for $C_i$, then $S_i \cap X = \emptyset$.

Note that, because all vertices in $X$ have priority smaller than $\lambda$, $X \cap V_i = X \setminus \text{Attr}(G[X],X \setminus V_i,\overline{\sigma})$. In particular, we have that $X \cap V_i$ is a $\overline{\sigma}$-trap in $G[X]$. Since ParAlg is nice with traps, this implies that ParAlg$(G[X \cap V_i],\sigma) = \emptyset$.

Since $X$ is a $\sigma$-trap in $G$ and $X \cap S_i = \emptyset$, we get that $X$ is a $\sigma$-trap in $G[V \setminus S_i]$. By definition, $V_i$ is the complement of a $\overline\sigma$ attractor, so $V_i$ is a $\overline\sigma$-trap in $G[V \setminus S_i]$. Since $X$ is a $\sigma$-trap and $V_i$ is a $\overline\sigma$-trap, we get that $X\cap V_i$ is a $\sigma$-trap in $G[V_i]$. Since $T_i$ is a $\overline\sigma$-trap in $G[V_i]$ and $X \cap V_i$ is a $\sigma$-trap in $G[V_i]$, we have that $X \cap V_i \cap T_i$ is a $\overline\sigma$ trap in $G[X \cap V_i]$. Therefore, since ParAlg is nice with traps, ParAlg$(G[X \cap V_i \cap T_i],\sigma) = \emptyset$. If $X \cap V_i \cap T_i$ were nonempty, then, since ParAlg is nice with traps and $X \cap V_i$ is a $\sigma$-trap in $G[V_i]$, we would have ParAlg$(G[X \cap V_i \cap T_i],\sigma) \neq \emptyset$, a contradiction. Therefore, we must have $X \cap V_i \cap T_i = X \cap T_i = \emptyset$.

Finally, since $C_{i+1} = S_i \cup T_i$, we have that $X \cap C_{i+1} = \emptyset$, as desired.
\end{proof}

\begin{lemma}
If $X$ is a $\sigma$-trap with largest vertex of priority $\lambda$, $\lambda$ has parity $\overline{\sigma}$, and ParAlg$(G[X],\sigma) = \emptyset$, then the largest vertices of $X$ are not contained in SeqAttr$(G,V,\sigma,\text{ParAlg})$.
\end{lemma}
\begin{proof}
Take $C_0=\emptyset$ and $W_0,W_1,\ldots$ to be the value of $W$ at the beginning of each ``WHILE'' loop in the execution of SeqAttr$(G,V,\sigma,\text{ParAlg})$. Take 
$$S_i = \max(W_i),$$
$$C_i = \text{GenAttr}(G,p(S_i),C_{i-1}\cup S_i,\sigma,\text{ParAlg}).$$

If $p(S_i) > \lambda$ we have, by maximality of $\lambda$, that $X \cap S_i = \emptyset$, and so by induction that $X \cap (C_{i-1} \cup S_i) = \emptyset$. By the previous lemma we get $X \cap C_i = \emptyset$. If $p(S_i) < \lambda$, then we have $\max(X) \cap S_i = \emptyset$ and so by induction $\max(X) \cap (S_i \cup C_{i-1}) = \emptyset$. Therefore, no vertices of $\max(X)$ can be added by the call to GenAttr and so $\max(X) \cap C_i = \emptyset$.
\end{proof}

\begin{theorem} 
TDA$(G,\sigma,\text{ParAlg})$ returns a set that is good for $\sigma$ with respect to ParAlg.
\end{theorem}
\begin{proof}
Take $S:= \text{TDA}(G,\sigma,\text{ParAlg})$. Then $S = \text{TDA}(G[S],\sigma,\text{ParAlg})$. We've observed before that $S$ is a $\overline{\sigma}$ trap whose largest vertex has priority of parity $\sigma$, so we may assume without loss of generality that $S=V$. By the previous lemma, there is no nonempty set $X$ that is a $\sigma$-trap with largest vertex of priority $\overline{\sigma}$ so that ParAlg$(G[X],\sigma)=\emptyset)$.
\end{proof}

This completes the proof of Theorem \ref{GeneralTDA}.

\subsubsection{Second Trap-Depth Algorithm}

We now define the second Trap-Depth Algorithm, TDA$(G,\sigma,k)$, and prove Theorem \ref{theorem_TDA}. We will define TDA$(G,\sigma,k)$ by recursively applying TDA$(G,\sigma,\text{ParAlg})$. In order to do so, we will need to know that TDA$(G,\sigma,\text{ParAlg})$ is nice with traps whenever ParAlg is.

\begin{lemma}
If ParAlg is nice with traps, then TDA$(G,\sigma,\text{ParAlg})$ is nice with traps.
\end{lemma}
\begin{proof}
The same argument used in proving Theorem \ref{intersectTrap} applies to show that, for any $\sigma$-trap $Y$, if TDA$(G,\sigma,\text{ParAlg}) \cap Y$ is nonempty, then TDA$(G[Y],\sigma,\text{ParAlg})$ is nonempty. The rest of the properties of being nice with traps follow from Theorem \ref{GeneralTDA}.
\end{proof}

We now define TDA$_k(G,\sigma) = \text{TDA}(G,\sigma,k)$ recursively. Define TDA$_0(G,\sigma)$ to be the algorithm that always returns the empty set. Note this is nice with traps. Given TDA$_k(G,\sigma)$, define TDA$_{k+1}(G,\sigma)$ by TDA$_{k+1}(G,\sigma)=\text{TDA}(G,\sigma,\text{TDA}_k)$. Inductively applying Theorem \ref{GeneralTDA} now proves Theorem \ref{theorem_TDA}.

Note that we had previously defined TDA$_1(G,\sigma)$. Recalling that 
$$\text{GenAttr}(G,\lambda,X,\sigma,\text{TDA}_0) = \text{SafeAttr}(G,\lambda,X,\sigma)$$ and that SafeAttr stabilizes its own output, it is easy to see that these two definitions match.

\subsubsection{Runtime}

\begin{lemma}
Let $T(n,m)$ be an upperbound on the runtime of ParAlg$(G,\sigma)$ for a graph $G$ on $n$ vertices and $m$ edges. Then the runtime of TDA$(G,\sigma,\text{ParAlg})$ on a graph on $n$ vertices and $m$ edges is at most $O(mn^2) + n^2T(n,m)$.
\end{lemma}
\begin{proof}
Consider first the Safe Attractor Algorithm. Since each iteration of the ``while'' loop increases the size of $C_{\text{cur}}$ or halts the algorithm, there will be at most $O(n)$ loops. If implemented carefully (in the same way that the regular Attractor is implemented) we may guarantee that each edge is only used a constant number of times and actually run the algorithm in $O(m+n) = O(m)$ time.

Next, consider the Generalized Safe Attractor Algorithm. Each iteration of the ``while'' loop increases the size of $C_{\text{cur}}$ or halts the algorithm. On top of calling ParAlg, the algorithm does $O(m)$ work for each loop (Restrict$(G,\lambda,\sigma)$ can be computed in linear time). If the algorithm runs $j$ ``while'' loops, it does work at most $(O(m)+ T(n,m)) \times j$.

The Sequential Attractor Algorithm has $C$ increasing every iteration or the algorithm halts. Note that each time a call to generalized attractor causes the generalized attractor to go through a ``while'' loop, a new vertex is added to $C$, so the total number of such loops done throughout the calls to generalized attractor is $n$, and so the total amount of work is at most $(O(m)+T(n,m)) \times n= O(mn) + nT(n,m)$.

In TDA we have $T_{\text{cur}}$ decreasing on each iteration or the algorithm halts, and so there are at most $n$ calls to SeqAttr, and on top of these only $O(m)$ work is done, and so we get $T(n,m) = O(mn^2) + n^2T(n,m)$.
\end{proof}

\begin{lemma}
Let $T(n,m,k)$ denote the runtime of TDA$(G,\sigma,k)$ for a graph $G$ on $n$ vertices and $m$ edges. Then $T(n,m,0) = O(1)$ and for $k>0$ we have $T(n,m,k)=O(mn^{2k-1})$.
\end{lemma}
\begin{proof}
We have $T(n,m,0)=O(1)$ since this algorithm always returns the empty set.

The same optimizations used in the computation of the Attractor and the Safe Attractor may be used to get a runtime of $O(m)$ in the case $k=1$ for the sequential attractor, that is for SeqAttr$_1$. In TDA we have $T_{\text{cur}}$ decreasing on each iteration or the algorithm halts, and so there are at most $n$ calls to SeqAttr$_1$, and on top of these only $O(m)$ work is done, and so we get $T(n,m,1)=O(mn)$.

For $k \geq 1$ by the previous lemma we have $T(n,m,k+1) \leq O(mn^2) + n^2T(n,m)$. This recurrence solves to $T(n,m,k) = O(mn^{2k-1})$, as desired.
\end{proof}

\section{Summary and Critical Remarks}

The theorems of the previous section show the promised characterization of TDA (Theorem \ref{theorem_TDA}):

\begin{quote}
TDA$(G,\sigma,k)$ returns the largest (possibly empty) set starting with which $\sigma$ can guarantee a win in at most $k$ moves in the trap-depth game on $G$.
\end{quote}

We have introduced Trap-Depth games (where the moves consist of choosing subsets of the graph rather than vertices/edges) and shown their close relationship with Muller games.
We have defined the trap-depth parameter and given algorithms for parity games for finding subsets of the winning regions whose runtime is bounded by an exponential in this trap-depth. Writing $d := |p(V)|$,
since the trap-depth of a parity game is at most $\left\lceil \frac{d}{2} \right\rceil$, the algorithm runs in time 
$O(mn^d)$. If one is only interested in the class of graphs with a bounded number of priorities, there are other options. The classical algorithm of Zielonka also runs in time $O(mn^d)$ (see \cite{GTW02}), but there are better algorithms: Jurdzinski's \cite{J00} algorithm achieves $O(dm\left(\frac{n}{\lfloor{\frac{d}{2}}\rfloor}\right)^{\lfloor{\frac{d}{2}}\rfloor})$, and the subexponential algorithm of \cite{JPZ06} achieves $n^{O(\sqrt{n})}$. Of course, the class of graphs of bounded trap depth is much more general than the class of graphs with a bounded number of priorities.

By Lemma~\ref{reducibility}, finding any nonempty subset of the winning region allows us to remove part of
the graph to get a smaller parity game that needs to be solved; thus, for example, Parity Games in which every subgame has bounded trap depth (such as those with a bounded number of priorities) may be completely solved in polynomial time, a generalization of the result that parity games with a bounded number of priorities may be solved in polynomial time.

Parity games are just one encoding of a class of Muller games. One may ask if there are others for which the characterization of Muller games we present is algorithmically useful. One possible encoding is called \emph{Explicit Muller Games}, where an enumeration of the sets winning for Red, i.e. of the set $\R$, is explicitly given as input. There is a known polynomial time algorithm for solving explicit Muller games \cite{F08}, but we may hope to obtain another algorithm using the characterization. If one could efficiently answer the following question, such an algorithm exists (note in the following question $(V,E,V_{\mbox{red}})$ are given explicitly): 

\begin{problem}
Given a Muller game $G$ and an explicit list $S_1,\ldots,S_k \subseteq V$, is there some polynomial time algorithm that determines if every red-trap $H$ contains one of the $S_i$ as a blue-subtrap?
\end{problem}

To see that the above would allow us to solve the problem, let an explicit Muller game $(V,E,V_{\mbox{red}},\R)$ be given. We will first prune $\R$ by removing any sets $R \in \R$ in which some vertex has no outgoing edges in $G[R]$ (these have no impact on the game). To determine if Red has a nonempty winning region, we will find the collection $W$ of sets in $\R$ from which Red will win the trap-depth game in which Blue goes first.

We will iteratively update $\R$ and $W$. Choose any minimal (under inclusion) set $R \in \R$. For each such set $R$ we determine if $G[R]$ contains any red-traps that do not contain as a blue-trap any set in $W \cup \{R\}$. If $G[R]$ has no such red-traps, then we add $R$ to $W$. In either case, we remove $R$ from $\R$ and iterate.

It is easy to argue that if in the trap-depth game the set of vertices is $X$ and it is Red's turn to move, then a blue-trap $Y$ in $G[X]$ is winning for Red if and only if $H$ is in $W$. To determine if Red has a non-empty winning region, we need only check if one of the sets in $W$ is a blue-trap in $G$.

\section*{Acknowledgments}
This work was partially supported by NSF grant DMS-0648208 at the Cornell REU, which are both gratefully acknowledged. Andrey Grinshpun is partially supported by the NPSC. Andrei Tarfulea is partially supported by the NSF GRFP. We warmly thank Alex Kruckman, James Worthington and Ben Zax for many stimulating discussions on an early part of this work, as well as Damian Niwinski for his comments. We also thank the anonymous referee, without whose comments reading this paper would be much less pleasant.


\end{document}